\documentclass{article}
\usepackage{ifthen}

\newcommand{\typeof}{1} %
\newcommand{\longversion}[1]{\ifthenelse{\equal{\typeof}{0}}{}{#1}\ignorespaces}
\newcommand{\shortversion}[1]{\ifthenelse{\equal{\typeof}{0}}{#1}{}\ignorespaces}
\newcommand{\longshortversion}[2]{\ifthenelse{\equal{\typeof}{0}}{#2}{#1}}

\usepackage{bussproofs}
\EnableBpAbbreviations

\usepackage{stmaryrd}

\title{The Lambda Calculus is Quantifiable} 

\author{Valentin Maestracci,\\ {\small I2M, Universit\'e d'Aix-Marseille  }
\and 
Paolo Pistone,\\ {\small LIP, Universit\'e Claude Bernard Lyon 1}}

 \usepackage{appendix}
\usepackage{stackrel}
 \usepackage[T1]{fontenc}
\usepackage{graphicx}
\usepackage{tabularx}
\usepackage{amsfonts}
\usepackage{cmll}
\usepackage{amsthm}
\usepackage{amsmath}
\usepackage{proof}
\usepackage{mathdots}
\usepackage{hyperref}
\usepackage{amssymb}
\usepackage{color}
\usepackage{thmtools} 
\usepackage{thm-restate}
\usepackage{xspace}
\usepackage{titlesec}
\titlespacing*{\section}{0ex}{1ex}{1.5ex}
\titlespacing*{\paragraph}{0ex}{1ex}{1.3ex}
\titlespacing*{\optionalsubsection}{0ex}{1ex}{1ex}

\newtheorem{example}{Example}[section]
\newtheorem{definition}{Definition}[section]

\newtheorem{remark}{Remark}[section]
\newtheorem{theorem}{Theorem}[section]

\newtheorem{lemma}[theorem]{Lemma}

\newtheorem{proposition}[theorem]{Proposition}

\newcommand{\modd}[1]{\llbracket#1\rrbracket }

\newcommand{\B}[1]{\mathbf{#1}}

\newcommand{\model}[1]{\modd{#1}}

\newcommand{\TT}[1]{\mathtt{#1}}

\newcommand{\C}[1]{\mathcal{#1}}
\newcommand{\BB}[1]{\mathbb{#1}}

\newcommand{\To}{\Rightarrow}

\usepackage{tikz-cd}
\tikzcdset{scale cd/.style={every label/.append style={scale=#1},
    cells={nodes={scale=#1}}}}
\usepackage{witharrows}
\newcommand\twoheaduparrow{\mathrel{\rotatebox{90}{$\twoheadrightarrow$}}}

\newcommand{\intv}{\mathbf{I}(\mathbb{R})}

\begin{document}

\maketitle

\begin{abstract}
In this paper we introduce several quantitative methods for the lambda-calculus based on partial metrics,
a well-studied variant of standard metric spaces that have been used to metrize non-Hausdorff topologies, like those arising from Scott domains. First, we study quantitative variants, based on program distances, of sensible equational theories for the $\lambda$-calculus, like those arising from B\"ohm trees and from the contextual preorder. Then, we introduce applicative distances capturing higher-order Scott topologies, including reflexive objects like the $D_\infty$ model. Finally, we 
 provide a quantitative insight on the well-known connection between the B\"ohm tree of a $\lambda$-term and its {Taylor expansion}, by showing that the latter can be presented as an isometric transformation.

\end{abstract}

\section{Introduction}

\subparagraph{Two notions of program approximation}
One of the fundamental goals of program semantics is to understand when two different programs compute \emph{the same} function. This is why, since its origins, the semantics of 
 the $\lambda$-calculus, the mathematical foundation for higher-order programming languages, has been focused on the problem of \emph{program equivalence}. Indeed, \emph{$\lambda$-theories}, the equational theories of the $\lambda$-calculus, constitute one of the pillars of the mathematical theory behind this much studied language, ranging from more operational theories, like $\beta$-equivalence, to more observational ones, like contextual equivalence. 

Actually, several well-known denotational models of the $\lambda$-calculus are not just the source for some $\lambda$-theory, but they also provide a \emph{topological} point of view on them: the interpretations of the $\lambda$-calculus via B\"ohm trees, Scott domains or the Taylor expansion, involve spaces whose objects can be seen as limits of "finite" approximants, as well as \emph{continuous} functions between such spaces, that is, functions commuting with such limits. In this way, the $\lambda$-theory induced by a topological model is associated with a notion of approximation, in the sense that a program is equivalent to another program whenever the net of finite approximants of the first converges to the second.
%

However, in general computer science, the approximation of a program is more commonly thought as the fact of computing values which are \emph{close} (possibly \emph{up to} some probability) to those produced by the program itself. By the way,
 the replacement of computationally expensive algorithms by more efficient, but somehow inaccurate, ones, is pervasive in all domains involving probabilistic or numerical methods. 
This has motivated, in the last few years, a rise of interest towards semantic approaches to functional languages focused, rather than on program equivalence, on notions of \emph{program similarity} \cite{Reed_2010, Gaboardi2017, DalLago2017, dallago, Hoshino2023}. In these approaches, each type is endowed with a \emph{pseudo-metric}, measuring the amount to which two programs behave in a similar, although non necessarily equivalent, way, and thus providing ways to estimate the errors produced by approximated optimization methods. At the same time, since any pseudo-metric induces an equational theory over programs, namely the one formed by all the pairs of programs which are at \emph{no} distance the one from the other, this approach
can be seen as a way to enrich, or "topologize", well-established notions of program equivalence.

\subparagraph{Quantifying $\lambda$-theories via partial metrics}

In a sense, the overall goal of this paper is to reconcile these two, apparently different, ways to look at program approximations, by developing metric counterparts to well-established methods for the $\lambda$-calculus, thus providing  ways to enrich $\lambda$-theories with notions of program similarity.

One reason why one could wish to approximate $\lambda$-theories by metrics is {computational}: while equational theories are generally undecidable, equivalences and, as we'll see, distances of finite approximants can often be computed effectively. Could one thus express the equivalence between two terms as the fact that the distance between their respective approximants gets closer and closer to zero?
This amounts to requiring that the limits in the topology $T_1$ generating the $\lambda$-theory are \emph{also} limits for the topology $T_2$ generated by some program pseudo-metric. In other words, that $T_1$ is \emph{finer} than $T_2$.

At the same time, since program metrics are generally undecidable as well, could the distances between two programs be themselves approximated by looking at the (computable) distances between their approximants? 
This amounts to requiring, conversely, that the metric limits, that is, the limits in $T_2$, are \emph{also} limits in the topology $T_1$ inducing the $\lambda$-theory. In other words, that $T_2$ is \emph{finer} than $T_1$.

All this sums up to the following question: can we make the topology arising from the semantics and the topology arising from the metric \emph{coincide}? 
At first, one would tend to answer no: for instance, while the topology of a metric space is \emph{always} Hausdorff, the topologies arising from the semantics of the $\lambda$-calculus (e.g.~Scott domains) are not. Nevertheless, there is a well-known reply to this answer, namely \emph{partial metrics} \cite{matthews, Bukatin1997, Stubbe2018, Waskiewicz2001,Schellekens2003, Valero2011}, a well-studied variant of standard metrics developed in connection with ideas from program semantics. A partial metric differs from a standard metric in that the self-distances $p(x,x)$ need \emph{not} be zero; correspondingly, one has a \emph{stronger} triangular law of the form $p(x,y)\leq p(x,z)+p(z,y)-p(z,z)$, taking into account the self-distance of the middle point $z$. As a consequence, distinct points will \emph{not} have disjoint neighborhoods, as soon as the self-distance of one makes it "too thick", so to say, to separate it from the other.

In fact, \emph{any} continuous domain with a countable basis is \emph{quantifiable} (a term we borrow from \cite{Schellekens2003}) by a partial metric. This means that its Scott topology does coincide with the topology induced by the metric \cite{Bukatin1997, ONeill, Waskiewicz2001, Schellekens2003, Smyth2006}, so that the limits in the Scott topology agree with the metric limits and viceversa.

While the quantification of domains via partial metrics has been well-known for a while, the application of such results   
to the study of higher-order languages has not been much explored so far. We do it in this paper: we introduce quantitative variants for well-known methods like B\"ohm trees, Scott domains and the Taylor expansion, based on partial metrics, at the same time providing ways to approximate their associated $\lambda$-theories.

\subparagraph{Contributions}

%
%
%
%

%

In this paper we show that several well-known approaches to the study of the $\lambda$-calculus can be \emph{quantified}, that is, enriched with metric reasoning on program similarity.
Our contributions can be summarized as follows:
\begin{itemize}
\item We introduce a partial metric variant of the notion of \emph{sensible} $\lambda$-theory \cite{Baren95} and we explore quantitative versions of well-known theories like those arising from B\"ohm trees and the contextual preorder.

\item We introduce \emph{applicative} partial metrics, and we illustrate their use to quantify higher-order Scott domains as well as reflexive objects, like Scott's model $D_\infty$. 
This opens the way to apply metric techniques to typed or non-typed higher-order languages.

\item Finally, we study the \emph{Taylor expansion} of $\lambda$-terms \cite{ER, Regnier2006, Barbarossa2019}, a powerful technique inspired by ideas from linear logic, and show that it can be presented as an isometric transformation from B\"ohm trees to sets of resource $\lambda$-terms, thus refining the well-known \emph{commutation theorem} \cite{Regnier2008}, that relates the corresponding $\lambda$-theories. 

\end{itemize}

\subparagraph{Outline}

In Section \ref{sec2} we recall basic notions about partial metric spaces.
In Section \ref{sec3} we introduce quantitative variants of sensible $\lambda$-theories. In Section \ref{sec4} we investigate the quantification of higher-order Scott domains via applicative distances, and in Section \ref{sec5} we apply these ideas to the quantification of reflexive objects. In Section \ref{sec6} we discuss the Taylor expansion.
Finally, in Section \ref{sec7} we indicate related work as well as a few future directions.

\section{Partial Metric Spaces}\label{sec2}


In this section we introduce partial metric spaces and we illustrate a few examples. 

\begin{definition}
	A function $p:X\times X\to [0,+\infty]$ is called a \emph{partial metric} (PM) when it satisfies the following axioms:
	\begin{description}
		\item[(P1)] $p(x,x) \leq p(x,y)$,		
		\item[(P2)] If $p(x,x) = p(x,y) = p(y,y)$ then $x=y$,
		\item[(P3)] $p(x,y) = p(y,x)$,
		\item[(P4)] $p(x,y)  \leq p(x,z) + p(z,y)-p(z,z)$.
	\end{description}
	$p$ is called a \emph{partial pseudo-metric} (PPM) when it satisfies P1,P3 and P4, and
	a \emph{partial ultra-metric} (PUM) when it satisfies P1, P3 and
	\begin{description}
		\item[(P4U)] $p(x,y)\leq \max\{p(x,z),p(z,y) \}$.
	\end{description}

\end{definition}

While in a standard (pseudo-)metric space each point is at distance $0$ from itself, 
condition P1 states that the distance of a point from itself is only required to be \emph{smaller} than its distance from any other point. Condition P2 adapts the usual separation condition 
$d(x,y)=0\To x=y$ to non-zero self-distances, and distinguishes PMs from PPMs. 
Condition P3 is the usual symmetry, while P4 is a strengthening of the triangular law of metric spaces, that also takes into account the possibly non-zero self-distance of the middle point $z$. 
P4U is as for standard ultra-metric spaces. Notice that P4U implies P4, so PUMs are indeed PPMs.
Notice that a PPM (resp.~a PUM, a PM) $p$ always induces a pseudo-metric (resp.~a ultra-metric, a metric) by the formula  $d_p(x,y):=2p(x,y)-p(x,x)-p(y,y)$.

A PPM $p$ induces a preorder on $X$ defined by $x \leq_p y$ iff $p(x,y) \leq p(x,x)$.
Notice that this implies by P1 that $p(x,y) = p(x,x)$. When $p$ is a PM the preorder $\leq_p$ is indeed an order. 	
With respect to this preorder, $p$ is \emph{antimonotonic} in the sense that $x\leq_p x'$ implies $p(x',y)\leq p(x,y)$.
Intuitively, the higher points are those with smaller self-distance.

The symmetrization of the preorder $\leq_p$ yields an equivalence relation $\simeq_p$.
 In the next section we will indeed explore the use of partial metrics as ways of approximating preorders or equivalence relations  on $\lambda$-terms.
We will say that a PPM $p$ \emph{quantifies} an order (resp.~an equivalence) relation over $X$ when this relation coincides with $\leq_p$ (resp.~$\simeq_p$).

Let us now talk about the topology induced by a PPM. 
\begin{definition}[open balls, topology]
Let $p$ be a PPM on $X$. For any $x\in X$ and $\epsilon\in (0,+\infty)$, the \emph{open ball of center $x$ and radius $\epsilon$} is the set $B_\epsilon^p(x)=\{ y\in X\mid p(y,x)< p(x,x)+\epsilon\}$.
The \emph{topology of $p$}, noted $\C O_p(X)$, is formed by all subsets $U\subseteq X$ which are unions of open balls.
\end{definition}

Recall that, by P1, the distance between two points $x,y$ is always greater or equal than the self-distances $p(x,x),p(y,y)$. 
We could equivalently define open balls as for standard metric spaces, i.e.~$B_\epsilon^p(x)=\{ y\in X\mid p(y,x)< \epsilon\}$, but this would make $B_\epsilon^p(x)$ \emph{empty} whenever $\epsilon\leq p(x,x)$. Open balls are upper: 
if $y\in B_\epsilon^p(x)$ and $y\leq_p y'$, by antimonotonicity we deduce 
$p(y',x)\leq p(y,x)<p(x,x)+\epsilon$, whence $y'\in B_\epsilon^p(x)$.
As a consequence, all open sets $U\in \C O_p(X)$ are upper.

Contrarily to standard metric spaces, the topology $\C O_p(X)$ is not in general Hausdorff: suppose $x,y$ are distinct points such that $x\leq_p y$; since any open set containing $x$ must also contain $y$, there can be no \emph{disjoint} open sets $U,V$ such that $x\in U$ and $y\in V$.  In some cases, as we'll see, $\C O_p(X)$ may coincide with the Scott topology induced by the order $\leq_p$.

In Sections 4 and 5 we will explore the use of partial metrics as ways of approximating (Scott) topologies on $\lambda$-terms.
We will say that a PPM $p$ \emph{quantifies} a topology $\C O(X)$ over $X$ when $\C O(X)=\C O_p(X)$.

\emph{Continuous} functions between PPMs can be defined in the usual topological sense:
given PPMs $p,p'$, respectively on $X$ and $X'$, a function $f:X\to X'$ is \emph{$p,p'$-continuous} when 
 $f^{-1}$ sends open sets in $\C O_{p'}(X')$ onto open sets in $\C O_p(X)$.
 There is an equivalent $\epsilon/\delta$-definition: $f$ is $p,p'$-continuous if for all $x\in X$ and $\epsilon>0$, there exists $\delta>0$ such that $f(B_\delta^p(x))\subseteq B_\epsilon^{p'}(f(x))$.

We compare different PPMs on a set $X$ by relating the associated topologies:
\begin{definition}
Given two PPMs $p,p'$ on $X$, we say that $p$ is \emph{finer than} $p'$ (noted $p\sqsubseteq p'$) when 
$\C O_{p'}(X)\subseteq \C O_p(X)$.
\end{definition}
Equivalently, $p\sqsubseteq p'$ when 
the identity map $\mathrm{id}_X:X\to X$ is $p,p'$-continuous, i.e.~every open $p'$-ball contains an open $p$-ball around any of its points. 

%


We conclude this short presentation with a few examples.

%
%
%

%
%
%
%
%
%
%

\begin{example}[Sierpinski space]
The simplest example of a non-Hausdorff topology that is quantified by a partial metric is the \emph{Sierpinski space} $S=\{0, 1\}$, with the Scott topology $\C O_\sigma(S)=\{\emptyset, \{1\}, \{0,1\}\}$.  Define the PM $s$ on $S$ by $s(0,0)=s(0,1)=1$ and $s(1,1)=0$. Notice how this implies $0 \leq_s 1$.
Since $0$ has self-distance $1$, the unique open balls are indeed $\emptyset$, $\{1\}$ and $\{0,1\}$, that is, $\C O_\sigma(S)=\C O_s(S)$.
\end{example}

\begin{example}[Intervals]\label{interval}
	The closed intervals of $\mathbb{R}$, noted $\intv$, admit the PM $p_{\mathrm{int}}(I_1, I_2) := \inf \{ |b-a| \mid I_1 \cup I_2 \subseteq [a, b]\}$, which is the size of the smallest interval containing $I_1$ and $I_2$.
	The order defined by the metric here is intuitive, it is reverse inclusion/the Scott information order: $I \leq_{p_{int}} J$ iff $p_{\mathrm{int}}(I,J) \leq p_{\mathrm{int}}(I,I)$ iff $J \subseteq I$. The more information one has, the higher.
	This example explains the choice of the word "partial": an interval, in term of Scott topology, represents an information on a partial execution: we have yet to discover the precise real number that we are computing. By contrast, the total elements will be those with self distance $0$ (the ones where $p$ behaves like a regular metric), i.e.~of the form $\{r\}$, a complete information, of a terminated execution.
\end{example}

\longversion{

\begin{restatable}{proposition}{sigmareal}\label{prop:sigmareal}
The map $p_{\mathrm{int}}(I_1,I_2)$ is a PM on $\intv$ and $\leq_{p_{\mathrm{int}}} = \supseteq$.
\end{restatable}
\begin{proof}
Notice that $p_{\mathrm{int}}(I_1,I_2)=\mathrm{diam}(I_1\cup I_2)$, where $\mathrm{diam}:\C P(\BB R)\to [0,+\infty]$ is the diameter function $\mathrm{diam}(A)=\sup\{|x-y| \ \mid \ x,y\in A\}$. $\mathrm{diam}$ is monotone and \emph{sub-modular}: if $I\cap J\neq \emptyset$, then $\mathrm{diam}(A\cup B)\leq \mathrm{diam}(A)+\mathrm{diam}(B)-\mathrm{diam}(A\land B)$.
Moreover, for any set $A\in \C P(\BB R)$, $\mathrm{diam}(A)=\mathrm{diam}(\overline A)$, where $\overline A$ is the smallest closed interval containing $A$. 

Notice that $\mathrm{diam}(A\cup B)=\mathrm{diam}(A)$ iff $B\subseteq \overline A$. 
Suppose $I_1,I_2$ are intervals such that 
$p_{\mathrm{int}} (I_1,I_2)=\mathrm{diam}(I_1\cup I_2)= \mathrm{diam}(I_1)$; this holds precisely when $I_2\subseteq \overline{I_1}=I_1$. This proves then that $\leq_{p_{\mathrm{int}}}$ coincides with the reverse incusion order $\supseteq$ on $\intv$.

Let us check now conditions P1-P4: 
P1 follows from the monotonicity of $\mathrm{diam}$; P2 follows from the argument above: if $p_{\mathrm{int}} (I_1,I_2)=\mathrm{diam}(I_1)=\mathrm{diam}(I_2)$, then $I_1\subseteq I_2$ and $I_2\subseteq I_1$, that is, $I_1=I_2$;
P3 is immediate; finally, for P4 we use the monotonicity and sub-modularity of $p_{\mathrm{int}}$:
\begin{align*}
p_{\mathrm{int}}(I_1,I_2)&= \mathrm{diam}(I_1\cup I_2) \\
&\leq \mathrm{diam}((I_1\cup I_3)\cup (I_3\cup I_2)) \\
&\leq \mathrm{diam}(I_1\cup I_3) + \mathrm{diam}(I_3\cup I_2) -\mathrm{diam}((I_1\cup I_3)\cap(I_3\cup I_2))\\
&\leq \mathrm{diam}(I_1\cup I_3) + \mathrm{diam}(I_3\cup I_2) -\mathrm{diam}(I_3)\\
&=p_{\mathrm{int}}(I_1,I_3)+p_{\mathrm{int}}(I_3,I_2)-p_{\mathrm{int}}(I_3,I_3).
\end{align*}
\end{proof}

}

\begin{example}[Labeled trees]\label{ex:tree}
 Let $\Sigma\mathrm{Tree}_{\leq\infty}$ be the set of (non necessarily finite) finitely branching $\Sigma$-labeled trees, where $\Sigma$ is a countable set of labels.
For any $\alpha\in\Sigma \mathrm{Tree}_{\leq\infty}$, let $|\alpha|\in \mathbb N\cup\{ \infty\}$ indicate the \emph{height} of $\alpha$. For any $n\in \mathbb N$, let $\alpha_n$ be the finite tree obtained by truncating all paths of $\alpha$ at length $n$, if $|\alpha| \geq n$, and be undefined otherwise. We write $\alpha_n \triangleq\beta_n $ to indicate that $\alpha_n$ and $\beta_n$ are both definite and equal, and $\alpha_n \not\triangleq\beta_n $ for its negation.
For any $\alpha, \beta\in\Sigma \mathrm{Tree}_{\leq\infty}$, define 
		$\mathrm{div}(\alpha,\beta):=\inf \{ n \mid \alpha_n \triangleq\beta_n \text{ and } \alpha_{n+1}\not\triangleq\beta_{n+1}    \}$.
		
		The standard tree (ultra-)metric $d_{\mathrm{tree}}$ is defined by $d(\alpha,\beta)=0$ if $\alpha=\beta$ and $2^{-\mathrm{div}(\alpha,\beta)}$ otherwise. We obtain instead a PUM by simply letting
		$p_{\mathrm{tree}}(\alpha,\beta):= 2^{-\mathrm{div}(\alpha,\beta)}$ (where it is intended that $2^{-\infty}=0$). 
For a finite tree $\alpha$, its self-distance is $p_{\mathrm{tree}}(\alpha, \alpha)=2^{-|\alpha|}$, while $p_{\mathrm{tree}}(\alpha, \alpha)=0$ holds iff $\alpha$ has infinite height. 
Also this case suggests that finite trees are seen as "partial" objects, while the infinite trees are the "total" ones.
Indeed, $p_{\mathrm{tree}}$, unlike $d_{\mathrm{tree}}$, quantifies the Scott topology on $ \Sigma\mathrm{Tree}_{\leq\infty}$ (see Section 4).
\end{example}

\longversion{

\begin{proposition}
$p_{\mathrm{tree}}$ is a PUM over $\Sigma\mathrm{Tree}_{\leq \infty}$ and $\leq_{p_{\mathrm{tree}}} =\leq$.
\end{proposition}
\begin{proof}
Let us check that $p_{\mathrm{tree}}(\alpha,\alpha)\leq p_{\mathrm{tree}}(\alpha,\beta)$. If $\alpha$ is a finite tree, then the highest value $n$ such that $|\alpha|,|\beta|\geq n$ and $\alpha_n=\beta_n$ must be smaller than $| \alpha|$, whence $p_{\mathrm{tree}}(\alpha,\beta)=2^{-n}\geq 2^{|\alpha|}=p_{\mathrm{tree}}(\alpha,\alpha)$. 
If $T$ is infinite, then $p_{\mathrm{tree}}(\alpha,\alpha)=0$ so the claim is immediately true.

Symmetry $ p_{\mathrm{tree}}(\alpha,\beta)= p_{\mathrm{tree}}(\beta,\alpha)$ is immediate by definition.

Let us show that $p_{\mathrm{tree}}(\alpha,\beta)\leq p_{\mathrm{tree}}(\alpha,\alpha)$ implies $\alpha\leq \beta$: if $\alpha$ is finite, then 
$p_{\mathrm{tree}}(\alpha,\beta)\leq p_{\mathrm{tree}}(\alpha,\alpha)$ implies that for all $ n\leq |\alpha|$, $\alpha_n=\beta_n$, which forces indeed $\alpha\leq \beta$; if $\alpha$ is infinite, then $p_{\mathrm{tree}}(\alpha,\beta)\leq p_{\mathrm{tree}}(\alpha,\alpha)$ implies $p_{\mathrm{tree}}(\alpha,\beta)=0$, which in turn implies that $\beta$ is infinite as well and for all $n\in \BB N$ $\alpha_n=\beta_n$. This implies then $T=U$ and thus in particular $\alpha\leq \beta$.

Conversely, if $\alpha\leq \beta$, then for all $n$ such that $|\alpha|,|\beta|\geq n$, it must hold $\alpha_n=\beta_n$, and this implies 
$p_{\mathrm{tree}}(\alpha,\beta)\leq p_{\mathrm{tree}}(\alpha,\alpha)$.

We can thus conclude that $\leq_{p_{\mathrm{tree}}} =\leq$. Notice that, together with symmetry, this also implies that if 
$p_{\mathrm{tree}}(\alpha,\beta)\leq p_{\mathrm{tree}}(\alpha,\alpha), p_{\mathrm{tree}}(\beta,\beta)$, then $\alpha\leq \beta$ and $\beta\leq \alpha$ both hold, and thus $\alpha=\beta$.

Finally, let us prove $p_{\mathrm{tree}}(\alpha,\beta)\leq \max\{p_{\mathrm{tree}}(\alpha,\gamma),p_{\mathrm{tree}}(\gamma,\beta)\}$: suppose that $|\alpha|,|\gamma|\geq n$ and $\alpha_n=\gamma_n$ and that $2^{-n}=p_{\mathrm{tree}}(\alpha,\gamma)\geq  p_{\mathrm{tree}}(\gamma,\beta)$. This implies than that also $|\beta|\geq n$ and $\gamma_n=\beta_n$ must hold. We conclude then that $\alpha_n=\beta_n$ holds and thus $p_{\mathrm{tree}}(\alpha,\beta)\leq 2^{-n}=p_{\mathrm{tree}}(\alpha,\gamma)=\max\{p_{\mathrm{tree}}(\alpha,\gamma),p_{\mathrm{tree}}(\gamma,\beta)\}$. We can argue in a similar way if we assume $p_{\mathrm{tree}}(\alpha,\gamma)<  p_{\mathrm{tree}}(\gamma,\beta)$.
The only remaining case is if $p_{\mathrm{tree}}(\alpha,\gamma)=p_{\mathrm{tree}}(\gamma,\beta)=0$: then $\alpha=\gamma=\beta$ and thus $p_{\mathrm{tree}}(\alpha,\beta)=0\leq \max\{p_{\mathrm{tree}}(\alpha,\gamma),p_{\mathrm{tree}}(\gamma,\beta)\}$.
\end{proof}

}

\section{Quantifying $\lambda$-Theories}\label{sec3}


%
%
%
%
%
%
%
%
%
%
%

In this section we introduce quantitative variants, based on partial metrics, of sensible $\lambda$-theories that arise from well-studied models of the untyped lambda-calculus, that is, the theory of B\"ohm trees and the theory of contextual equivalence. Moreover, we lift several properties of such equational theories to the corresponding notion of program similarity.

\subparagraph{$\lambda$-PPMs}

Let us first recall the standard notion of $\lambda$-theory \cite{Baren95}. 
\begin{definition}
A \emph{$\lambda$-theory} $T$ is an equivalence relation $\simeq_T$ on the set $\Lambda$ of all $\lambda$-terms satisfying the rules below:
\begin{description}
\item[(congr1)] $M\simeq_T N \To MP\simeq_T NP$,
\item[(congr2)] $M\simeq_T N \To PM\simeq_T PN$,
\item[($\xi$)] $M\simeq_T N \To \lambda x.M\simeq_T \lambda x.N$,
\item[($\beta$)] $(\lambda x.M)N \simeq_T M[N/x]$.
\end{description}
A $\lambda$-theory $T$ is said \emph{extensional} when it furthermore satisfies the rule $(\eta)$:
\begin{description}
\item[($\eta$)] $M\simeq_T \lambda x.Mx$.
\end{description}
 A $\lambda$-theory $T$ is said   
 \emph{sensible} when it equates all unsolvable terms and does not equate a solvable and an unsolvable term. 
\end{definition} 
 Notice that a sensible theory $T$ must be consistent: it cannot equate \emph{all} terms.

A $\lambda$-theory may either arise from an operational theory (e.g.~$\beta$- and $\beta\eta$-reduction) or be induced by a model (as the theory formed by all equations between terms that are interpreted by the same entity in the model).
While there exists a continuum of different $\lambda$-theories, beyond the theories of $\beta$ and $\beta\eta$-equivalence (respectively, the smallest $\lambda$-theory and the smallest extensional $\lambda$-theory), very few theories have been studied in depth. Indeed, all most common denotational models of the untyped $\lambda$-calculus induce one of the two sensible theories $\mathcal B$, and $\mathcal H^*$, that we consider below.

We now introduce a quantitative variant of $\lambda$-theories.  
Let us first recall that a point $x$ in a topological space $X$ is said \emph{generic} when its closure is $X$ or, equivalently, all its neighborhoods are dense in $X$. For instance, $0$ is generic in the Sierpinski space $S$. In the case of PPM we have the following:

%
%

\begin{lemma}
$x$ is generic in the topology $\C O_p(X)$ iff $x\leq_p y$ holds for all $y\in X$.
\end{lemma}
\begin{proof}
Call $x$ generic \emph{for $p$} if $x\leq_p y$ (that is, $p(y,x)= p(x,x)$) holds for all $y\in X$. $x$ is generic for $p$ iff the only open ball centered at $x$ is $X$: from $p(y,x)=p(x,x)$ it follows that for any $\epsilon>0$, $y\in B_\epsilon(x)$, that is, $B_\epsilon(x)=X$; conversely, if any open ball centered at $x$  is equal to $X$, then, for all $\epsilon>0$, $p(y,x)< p(x,x)+\epsilon$, which implies $p(y,x)\leq p(x,x)$ and thus $p(y,x)=p(x,x)$ by P1.

Now, if $x$ is generic for $p$, then any open set $U$ containing $x$ must contain some $B_\epsilon(x)$, which forces $U=X$, so $x$ is generic in $\C O_p(X)$. Conversely, if $x$ is generic in $\C O_p(X)$, then for any $\epsilon>0$, the closure of $B_\epsilon(x)$ is $X$. This implies that for all $\epsilon>0$, $p(y,x)\leq p(x,x)+\epsilon$, and thus that $p(y,x)=p(x,x)$, so $x$ is generic for $p$.
\end{proof}

\begin{remark}\label{rem:generic}
Generic points are indistinguishable: if $x$ and $y$ are both generic for $p$, then from $p(y,y)=p(x,y)=p(x,x)$ it follows that $x\simeq_p y$. Conversely,  if $x$ is generic and $y$ is not, then, $x\not\simeq_p y$: if $x\simeq_p y$ held, then, for all $z$, $p(y,z)\leq p(y,x)+p(x,z)-p(x,x)=p(y,x)+p(x,x)-p(x,x)=p(y,x)=p(y,y)$, so $y$ would be generic as well.
\end{remark}

\begin{definition}[$\lambda$-PPM]
A pseudo-partial metric $p$ over $\Lambda$ is called a \emph{$\lambda$-PPM} (resp.~an \emph{extensional $\lambda$-PPM})  if the following hold:
\begin{itemize}
\item $\simeq_p$ is a $\lambda$-theory (resp.~an extensional $\lambda$-theory);
\item all contexts $\mathtt C[-]$ correspond to $p$-continuous maps $\Lambda\to \Lambda$.
\end{itemize}
$p$ is called \emph{sensible} if all unsolvable terms are generic while no solvable term is.
\end{definition}

Observe that we do not require contexts to be \emph{non-expansive} (or $1$-Lipschitz), as in other standard metric approaches \cite{Reed_2010, Gaboardi2017, Honsell2022}, but just continuous. 
Also notice that, by Remark \ref{rem:generic}, a sensible PPM $p$ must satisfy $M\simeq_p N$ for all unsolvable terms $M,N $, and $M\not\simeq_pN$ for $M$ unsolvable and $N$ solvable: the associated $\lambda$-theory $\simeq_p$ is thus sensible. 

In the rest of this section we introduce $\lambda$-PPMs quantifying the $\lambda$-theories $\C B$ and $\C H^*$.

\subparagraph{B\"ohm Trees}

The interpretation of $\lambda$-terms as B\"ohm trees is one of the fundamental tools in the $\lambda$-calculus.
The B\"ohm tree $\C B(M)$ of a $\lambda$-term $M$ is a $(\Lambda\cup\{\bot\})$-labeled tree defined \emph{co-inductively} as follows:
\begin{itemize}
\item if $M$ reduces to $\lambda x_1.\dots.\lambda x_m.xM_1\dots M_n$, then $\mathcal B(M) $ has a root with label $\lambda x_1.\dots.\lambda x_m.x$ and $n$ subtrees
$\mathcal B(M_1),\dots, \mathcal B(M_n)$;
\item otherwise, $\mathcal B(M)$ only consists of the root, with label $\bot$. 

\end{itemize} 

An alternative presentation of $\C B(M)$ is via \emph{partial terms}, which are $\lambda$-terms in normal form, enriched with the constant $\bot$ and rules $\lambda x. \bot \to \bot$, $\bot M \to \bot$. We note these partial terms $A,B,\dots$. The set 
$\C A$ of partial terms is ordered by the contextual closure $\preceq$ of the relation generated by $\bot \preceq A$, for all $A\in \C A$. Partial terms correspond straightforwardly to \emph{finite} B\"ohm trees.

For any $\lambda$-term $M$, let the partial term $M_{\C A}$ be defined \emph{inductively} as follows: 
$M_{\C A}=\lambda \vec x. y (M_1)_{\C A}\dots  (M_n)_{\C A}$ if $M=\lambda \vec x.yM_1\dots M_n$, and 
$M_{\C A}=\bot$ if $M=\lambda \vec x.(\lambda y.P)M_1\dots M_{n+1}$. 
Let $A\leq M$ whenever $M$ $\beta$-reduces to $M'$ with $A\preceq M'_{\C A}$. We then let $\C B(M)=\{ A\mid A\leq M\}$. Observe that $\C B(M)$ can be seen at the same time as a tree under the relation $\leq$, and the standard tree ordering $\C B(M)\preceq \C B(N)$ holds precisely when $\C B(M)$ is included in $\C B(N)$.
%
%
%

The $\lambda$-theory $\mathcal B$ contains all equations $M\simeq_{\C B} N$, where $\C B(M)=\C B(N)$. $\mathcal B$ is sensible but non-extensional (as e.g.~$\mathcal B(\lambda x.x)\neq \mathcal B(\lambda x.\lambda y.xy)$).

We now introduce the corresponding $\lambda$-PPM: we measure the distance between $\lambda$-terms by comparing their B\"ohm trees via the tree partial metric.

\begin{definition}[B\"ohm partial metric]
For any two $\lambda$-terms $M,N$, let
\begin{align*}
p_{\mathrm{B\"ohm}}(M,N) &:= p_{\mathrm{tree}}(\mathcal B(M), \mathcal B(N)).
\end{align*}
\end{definition}

Observe that 
$p_{\mathrm{B\"ohm}}(M,M)=0$ iff $\mathcal B(M)$ is infinite. 
It is not difficult to check that $p_{\mathrm{B\"ohm}}$ captures the theory $\C B$:

\begin{restatable}{proposition}{bohm}
$M \leq_{p_{\mathrm{B\"ohm}}} N$ iff $\mathcal B(M)\leq \mathcal B(N)$, and thus 
$M\simeq_{p_{\mathrm{B\"ohm}}} N$ iff $M\simeq_{\mathcal B} N$.
\end{restatable}
\longversion{
\begin{proof}
Follows from Proposition 
\ref{prop:ptree}, proved below.
\end{proof}
}
 As discussed in Section 4, 
$p_{\mathrm{B\"ohm}}$ captures the Scott topology of B\"ohm trees. This proves that contexts are continuous, and thus that $p_{\mathrm{B\"ohm}}$ is a $\lambda$-PPM. Moreover, since $p_{\mathrm{tree}}(\bot, \alpha)=1$, the unsolvable terms are
generic, while, for any solvable term $M$, $p_{\mathrm{B\"ohm}}(M,M)<1$ and thus, for any $\epsilon < 1-p_{\mathrm{B\"ohm}}(M,M)$, the open ball $B_{\epsilon}^{p_{\text{B\"ohm}}}(M)$ does not contain the term $\lambda x.M$
(since $p_{\mathrm{B\"ohm}}(M,\lambda x.M)=1> p_{\mathrm{B\"ohm}}(M,M)+\epsilon$).

\begin{remark}
While the theory $\C B$ is $\Pi^0_2$-complete, the distances $p_{\mathrm{tree}}(A,B)$ are effectively computable whenever $A,B$ are \emph{finite} trees (equivalently, partial terms). Moreover, to check that $p_{\mathrm{B\"ohm}}(M,N)<\epsilon$, it is necessary and sufficient to find approximants 
$A\leq M$ and $B\leq N$ such that $p_{\mathrm{tree}}(A,B)<\epsilon$. 
\end{remark}

%
%
%

\subparagraph{Contextual equivalence}

%
%

We now consider the theory arising from \emph{contextual equivalence}. 
Let $M\sqsubseteq_{\mathrm{ctx}} N$ if for all context $\mathtt C[-]$, if $\mathtt C[M]$ is solvable, then $\mathtt C[N]$ is solvable. 
The theory $\C H^*$ contains all equations $M \simeq_{H^*} N$ where $M\sqsubseteq_{\mathrm{ctx}}N$ and $N\sqsubseteq_{\mathrm{ctx}}M$ both hold. 
It is extensional and sensible, and is indeed the \emph{maximum} sensible theory.

To quantify $\C H^*$ we define the following distance:
\begin{definition}[contextual partial metric]
For all terms $M,N$, we define 
$$
p_{\mathrm{ctx}}(M,N)= \sum_{n =0}^\infty \left \{\frac{1}{2^n} \ \Big \vert \  
\mathtt C_n[M]\text{ is unsolvable or }
\mathtt C_n[N]\text{ is unsolvable}
\right\},
$$
where $(\mathtt C_n[-])_{n\in \mathbb N}$ is an enumeration of all contexts. 
\end{definition}

The distance $p_{\mathrm{ctx}}(M,N)$ intuitively counts all contexts $\mathtt C_n[-]$ that fail on either $M$ or $N$. In particular, the self-distance $p_{\mathrm{ctx}}(M,M)$ counts the contexts that fail on $M$.

The following result shows that $p_{\mathrm{ctx}}$ captures the contextual preorder:
\begin{proposition}
$M \leq_{p_{\mathrm{ctx}}} N$ iff $M \sqsubseteq_{\mathrm{ctx}}N$, and thus 
$M\simeq_{p_{\mathrm{ctx}}} N$ iff $M \simeq_{\mathcal H^*} N $.
\end{proposition}

For the result above, the choice of the enumeration is irrelevant, as is the choice of the weights $\frac{1}{2^n}$, which could be replaced by arbitrary weights $\theta_n$ summing up to $1$. 

\begin{remark}\label{rem:ctx}
Contrarily to contextual equivalence, which is $\Pi^0_2$-complete as well, to check that $N\in B_\epsilon^{p_{\mathrm{ctx}}}(M)$ one does not need to look at the behavior of $M$ and $N$ under \emph{all} contexts. Intuitively, $B_\epsilon^{p_ctx}(M)$ contains all those programs that behave like $M$ on certain \emph{finitely many} contexts. 
Indeed, $p_{\mathrm{ctx}}(M,N)<p_{\mathrm{ctx}}(M,M)+\epsilon$ means that the contexts on which $M$ does converge and $N$ does not sum up to some value strictly smaller than $\epsilon$. This is true iff $N$ converges on those {finitely many} contexts $\TT C_i[-]$, where $2^{-(i+1)}\leq\epsilon$, on which $M$ converges. 
%
%
\end{remark}

\begin{restatable}{proposition}{sensiblectx}
$p_{\mathrm{ctx}}$ is a sensible extensional $\lambda$-PPM.
\end{restatable}
\begin{proof}
Let us show that contexts yield continuous maps. Take a term $M$, $\epsilon>0$ and a context $\TT C$. We need to find some $\delta>0$ such that for all $P\in B_\delta^{p_{\mathrm{ctx}}}(M)$, $\TT C[P]\in B_\epsilon^{p_{\mathrm{ctx}}}(\TT C[M])$. 
By Remark \ref{rem:ctx} there exists a \emph{finite} number of contexts $\TT C_1,\dots, \TT C_k$ such that 
$\TT C_i[\TT C[M]]$ is solvable and 
$N\in B_\epsilon^{p_{\mathrm{ctx}}}(\TT C[M])$ iff 
$\TT C_i[N]$ is solvable for $i=1,\dots, k$. 
Take $m$ such that for all $i=1,\dots, k$, the context $\TT C_i[\TT C[-]]$ has an index smaller than $m$, and let $\delta=2^{-m}$. Notice that $\TT C_i[\TT C[M]]$ is solvable. Moreover, for any term $P$, if $P\in B_\delta^{p_{\mathrm{ctx}}}(M)$, then 
 $\TT C_i[\TT C[P]]$ must be solvable for all $i=1,\dots, m$. This implies then that $\TT C[P]\in B_\epsilon^{p_{\mathrm{ctx}}}(\TT C[M])$, as desired. 

The sensibility of  
$p_{\mathrm{ctx}}$ essentially follows from the well-known \emph{genericity lemma} \cite{Baren95, Arrial2024}: if $\TT C[M]$ is solvable, where $M$ is unsolvable, then $\TT C[N]$ must be solvable \emph{for all} $N$; this implies that for any unsolvable $M$, and for any term $N$, $p_{\mathrm{ctx}}(M,N)=p_{\mathrm{ctx}}(M,M)$, so $M$ is generic in $p_{\mathrm{ctx}}$. 
Conversely, if $M$ is solvable, then, for any unsolvable term $N$, one can easily construct a context $\TT C$ such that $\TT C[M]$ reduces to $\lambda x.x$ and $\TT C[N]$ diverges. This allows us to conclude that $p_{\mathrm{ctx}}(M,N)>p_{\mathrm{ctx}}(M,M)$, and thus that $M$ is not generic in $p$. 
\end{proof}

Similarly to the $\lambda$-theory $\C H^*$, the $\lambda$-PPM $p_{\mathrm{ctx}}$ is \emph{maximum} among sensible $\lambda$-PPMs.

\begin{restatable}{proposition}{coarsest}\label{prop:coarsest}
For any sensible $\lambda$-PPM $p$, $p\sqsubseteq p_{\mathrm{ctx}}$.
\end{restatable}
\begin{proof}
Let $p$ be a sensible $\lambda$-ppm. Consider a term $M$ and $\epsilon>0$. We must find $\delta>0$ such that
$B_\delta^p(M)\subseteq B_\epsilon^{p_{\mathrm{ctx}}}(M)$. 
By Remark \ref{rem:ctx} there exists a finite number of contexts $\TT C_1,\dots, \TT C_k$ such that 
$\TT C_i[M]$ is solvable and 
$N\in B_\epsilon^{p_{\mathrm{ctx}}}(M)$ iff 
$\TT C_i[N]$ is solvable for $i=1,\dots, k$.  

Fix an $i\leq k$ and let $Q_i=\TT C_i[M]$. Since $Q_i$ is solvable and $p$ is sensible, we can find an open set $U_i$ containing $Q_i$ and \emph{not} containing any unsolvable term. Since $p$ is a $\lambda$-PPM, $\TT C_i$ corresponds to a continuous function, and thus $\TT C_i^{-1}(U_i)$ contains some open ball $B^p_{\delta_i}(M)$. 
Let $\delta=\min_i\delta_i$: if $P\in B^p_{\delta}(M)$, then for all $i=1,\dots, k$, $\TT C_i[P]\in U_i$, so it must be solvable.
We conclude then that $P\in B_\epsilon^{p_{\mathrm{ctx}}}(M)$.
\end{proof}

\begin{remark}\label{rem:ctxbohm}
That $p_{\mathrm{B\"ohm}}\sqsubset  p_{\mathrm{ctx}}$ can be easily seen directly:  
the elements of $B^{p_{\mathrm{ctx}}}_\epsilon(M)$ are those which converge on a finite number of contexts $\TT C_1,\dots, \TT C_k$ on which $M$ converges too (cf.~Remark \ref{rem:ctx}). For any such context $\TT C_i$, the convergence of $\TT C_i[M]$ to a head normal form only depends on the exploration of a \emph{finite} portion of $\C B(M)$, say up to height $m_i$. Letting $m=\max_i\{m_i\}$ and $\delta=2^{-m}$, we have then that $B^{p_{\mathrm{B\"ohm}}}_\delta(M)\subseteq B^{p_{\mathrm{ctx}}}_\epsilon(M)$.
The converse inclusion $ p_{\mathrm{ctx}}\sqsubseteq p_{\mathrm{B\"ohm}}$ cannot hold: any open ball $B^{p_{\mathrm{B\"ohm}}}_\epsilon(I)$ around $I=\lambda x.x$ that does not contain its $\eta$-expansion $\lambda x.\lambda y.xy$ contains \emph{no}
open $p_{\mathrm{ctx}}$-ball around $I$.  
\end{remark}

%
%
%

Other well-known characterizations of $\mathcal H^*$ exist, which suggest different ways to quantify this theory. One is in terms of the so-called \emph{Nakajima trees} (cf.~\cite{Baren95}, Ex.~19.4.4, p.~511): these are a variant of B\"ohm trees that are invariant under the $\eta$-rule. By adapting the tree partial metric one could then obtain another partial metric $p_{\mathrm{Nakajima}}$ that quantifies $\mathcal H^*$.

Moreover, the theory $\mathcal H^*$ is induced by a large class of denotational models of the $\lambda$-calculus (cf.~\cite{ManzoTesi}), including in particular the models based on Scott domains, that we discuss in Sections 4 and 5, or the relational model from \cite{Manzo2007}, to which the techniques illustrated in those sections can be easily adapted.
%
%

\section{Quantifying Scott Domains}\label{sec4}

As discussed in the introduction, the $\lambda$-theories like $\mathcal B$ or $\mathcal H^*$ are induced by topological models, based on Scott domains, which
provide notions of approximant for $\lambda$-terms.
In this section, after discussing the connection between partial metrics and Scott domains, we introduce applicative PPMs as a means to capture domains of Scott-continuous functions, and we illustrate how this leads to quantify  topological models of typed $\lambda$-calculi.


\subparagraph{Scott Domains via Partial Metrics}

Let us recall some basic terminology about dcpos and Scott domains.

A partially ordered set $(X,\leq)$ is a \emph{dcpo} (directed complete partial order) if all directed subsets of $X$ admit a least upper bound. The \emph{way below} relation $\ll$ on a dcpo is defined by 
$
x \ll y $ iff for all directed subset $\Delta\subseteq X$, $y\leq\bigvee \Delta$ implies $x\leq d$, for some $d\in \Delta$.
A point $x\in X$ is said \emph{compact} if $x\ll x$.  
A \emph{basis} for a dcpo $X$ is a subset $B\subseteq X$ such that for any $x\in X$, the set 
$\Delta=\{y\in B\mid y\ll x\}$ is directed and $x=\bigvee \Delta$. A dcpo is said \emph{continuous} if it has a basis and \emph{algebraic} if it has a basis formed of compact elements.
%
A \emph{domain} is a continuous dcpo with a countable basis.
A domain $X$ is \emph{bounded complete} if for any finite set $Y\subseteq_{\mathrm{fin}} X$, if an upper bound of $Y$ exists in $X$, then $\bigvee Y$ exists in $X$. A bounded complete and algebraic domain is called a \emph{Scott domain}. 

The \emph{Scott topology} $\mathcal O_\sigma(X)$ on a partially ordered set $(X,\leq)$ has open sets being upper subsets $U\subseteq X$ which are \emph{finitely accessible}: $x\in U$ implies $y\in U$ for some $y\ll x$. 
A function $f:X\to Y$ between dcpos is said \emph{continuous} iff $f$ is monotone and commutes with the lubs of directed subsets, that is, for all directed $\Delta\subseteq X$, $f(\bigvee \Delta)=\bigvee f(\Delta)$. This is equivalent to asking $f$ to be continuous, in the usual sense, with respect to the Scott topology.
The category of bounded complete domains and continuous functions is cartesian closed (cf.~\cite{Amadio1998}).


    We will say that a dcpo $(X, \leq)$ is \emph{quantified by a partial metric $p$} when its associated Scott topology is quantified by $p$, that is, when $\mathcal O_p(X)=\mathcal O_\sigma(X)$.

\longversion{
We start by discussing the examples from Section 2.

\begin{proposition}
$p_{\mathrm{int}}(I_1,I_2)$ quantifies $\intv$.
\end{proposition}
\begin{proof}
First observe that a countable basis for $\intv$ is formed by the rational closed intervals, and that $I\ll J$ holds iff 
$I=[a,b]$, $J=[a',b']$, with $a<a'$ and $b'<b$.

\begin{description}
\item[($\C O_\sigma(\intv)\supseteq \C O_{p_{\mathrm{int}}}(\intv)$)]
Let us show that open balls are finitely accessible. 
If $p_{\mathrm{int}}(J,I)<p_{\mathrm{int}}(I,I)+\epsilon=\mathrm{diam}(I)+\epsilon$, then we can factor $\epsilon=\theta+\delta$, with $\delta,\theta>0$ so that $p_{\mathrm{int}}(J,I)<p_{\mathrm{int}}(I,I)+\delta$; 
letting $J=[a,b]$, define $J'=[a-\frac{\theta}{2}, b+\frac{\theta}{2}]$, so that $J'\ll J$ and $p_{\mathrm{int}}(J',I)<p_{\mathrm{int}}(I,I)+\delta+\theta=p_{\mathrm{int}}(I,I)+\epsilon$; we conclude then $J'\in B_\epsilon(I)$.

\item[($\C O_\sigma(\intv)\subseteq \C O_{p_{\mathrm{int}}}(\intv)$)]
Let us show that the basic Scott open sets $\twoheaduparrow I =\{ J \mid I\subsetneq J\}$ are $p_{\mathrm{int}}$-open.
Let $I=[a,b]\ll [a',b']=J$, and let $\epsilon=\min\{ \frac{b-b'}{3}, \frac{a'-a}{3}  \}$.
Then, for all $J'\in B_\epsilon^{p_{\mathrm{int}}}(J)$, 
$p_{\mathrm{int}}(J',J)=\mathrm{diam}(J'\cup J)< \mathrm{diam}(J)+\epsilon$, which shows that $J'\subset [a'-\epsilon, b'+\epsilon]\gg [a,b]=I$. 
\end{description}
\end{proof}

\begin{restatable}{proposition}{sigmatree}\label{prop:ptree}
$p_{\mathrm{tree}}$ quantifies $\Sigma\mathrm{Tree}_{\leq\infty}$.
\end{restatable}
\begin{proof}
\begin{description}
\item[($\C O_\sigma(\C B)\supseteq \C O_{p_{\mathrm{tree}}}(\C B)$)]
Let us show that $p$-balls are finitely accessible. 
  Let $ \gamma\in B_{\epsilon}(\alpha)$; we look for some $\gamma'\ll\gamma$ such that $\gamma'\in B_\epsilon(\alpha)$. We can suppose w.l.o.g.~that $\epsilon=2^{-n}$. Then either $|\alpha|<n+1$, $|\gamma|< n+1$ or 
	$|\alpha|,|\gamma|> n$ and $\alpha_{n+1}\neq \gamma_{n+1}$. In the first case then $p_{\mathrm{tree}}(\alpha, \gamma_{n+1})=p_{\mathrm{tree}}(\alpha,\gamma)$, whence 
	$ \gamma_{n+1}\in B_{\epsilon}(\alpha)$ (and notice that $\gamma_{n+1}\ll \gamma$); in the second case $\gamma$ is finite and we can choose $\gamma':=\gamma$, since $\gamma\ll \gamma$; in the third case we have again $p_{\mathrm{tree}}(\alpha, \gamma_{n+1})$ and thus we can take
	$\gamma':=\gamma_{n+1}$.

\item[($\C O_\sigma(\C B)\subseteq \C O_{p_{\mathrm{tree}}}(\C B)$)]
Let us show that for all trees $\alpha\ll \beta$ we can find $n$ such that $B_{2^{-n}}(\beta)\subseteq \ \twoheaduparrow \alpha$. 

First observe that from $\alpha\ll\beta$ and the remark that finite trees are a basis of $\C B$ it follows that there exists finitely many finite trees $\beta_1,\dots,\beta_k$ such that $\alpha\leq \bigvee_i \beta_i \ll \beta$. This shows in particular that $\alpha$ must be a finite tree as well. We consider different cases:
	\begin{itemize}
	\item $|\beta |>|\alpha|$: then $\beta\in B_{2^{-(|\alpha|+1)}}(\beta)$ implies that $|\gamma|> |\alpha| $ and that $\gamma_n=\alpha_n$ holds for all $n\leq |\alpha|$. This implies then $\alpha\ll\gamma$.

	\item $|\beta|=|\alpha|$: then $\gamma\in B_{2^{-(|\beta|)}}(\beta)$ implies that $|\gamma|\geq |\alpha| $ and that $\gamma_n=\alpha_n$ holds for all $n\leq |\alpha|$. This implies in particular that $\beta=\gamma_{|\beta|}$. From this we deduce that $\alpha\ll\gamma_{|\beta|}$ and thus in particular $\alpha\ll \gamma$. 
	\end{itemize}
\end{description}
\end{proof}

}

\shortversion{
Two simple examples are the following (proofs are in the long version):
\begin{proposition}\label{prop:ptree}
Tthe interval dcpo $\intv$ is quantified by $p_{\mathrm{int}}$ (cf.~Example \ref{interval}).
The domain $\Sigma\mathrm{Tree}_{\leq\infty}$ of $\Sigma$-trees is quantified by $p_{\mathrm{tree}}$ (cf.~Example \ref{ex:tree}).
\end{proposition}
}

A very general result on the quantifiability of domains, mentioned in the Introduction, is the following:
\begin{theorem}[cf.~\cite{Schellekens2003}]\label{thm:full}
Let $(X,\leq)$ be a domain with a countable basis $(b_n)_{n\in \mathbb N}$, and let $\theta_n\in (0,1]$ be a sequence of weights such that $\sum_n^\infty\theta_n\leq 1$. Then $X$ is quantified by the partial metric 
$p^X_{b_n, \theta_n}(x,y)=\sum_{n \in N} \theta_n$, where $N := \left\{ n \  \vert \  b_n\not \ll x \ \text{or} \ b_n \not \ll y\right  \}$.
\end{theorem}

While Theorem \ref{thm:full} provides a general positive answer to the \emph{quantifiability} problem for domains, the practical usability of metrics like $p^X_{b_n, \theta_n}$ depends on whether the relation $b_n\ll x$ between a point and an approximant, and its negation, are computable (cf.~Remark \ref{rem:compu} below).

When a dcpo $X$ is quantified by a partial metric $p$, the order of $X$ coincides with the order induced by $p$.

\begin{restatable}{lemma}{Xleq}
Suppose the dcpo $(X,\leq)$ is quantified by $p$. Then $\leq$ coincides with $\leq_p$.
\end{restatable}
\begin{proof}
$\leq$ coincides with the \emph{specialization order} $x\leq^{\mathcal O_\sigma(X)}y \Leftrightarrow\forall U\in \mathcal O_\sigma(X)(x\in U\Rightarrow y\in U)$; similarly, $\leq_p$ coincides with the {specialization order} 
$x\leq^{\mathcal O_p(X)}y \Leftrightarrow\forall U\in \mathcal O_p(X)(x\in U\Rightarrow y\in U)$. From $\mathcal O_\sigma(X)=\mathcal O_p(X)$ we deduce that the two specialization orders coincide, and thus $\leq$ and $\leq_p$ as well.
\end{proof}
%

%
%
%
%
%
%

However, checking that a partial metric $p$ captures the order of the dcpo is \emph{not} in general enough to deduce that $p$ quantifies the dcpo (e.g.~the example in Remark \ref{rem:leqleqp}). 
The following proposition provides necessary (but not sufficient) conditions.


\begin{restatable}{proposition}{continuous}\label{prop:continuous}
Let $(X,\leq)$ be a continuous dcpo and $p$ a partial metric on $X$ such that $\leq$ coincides with $\leq_p$.Then the following conditions are equivalent:
\begin{enumerate}
\item $\C O_p(X)\subseteq \C O_\sigma(X)$;
\item open $p$-balls are finitely accessible;
\item $p$ is Scott-continuous (as a map towards the dcpo $([0,+\infty],\geq)$).
\end{enumerate}
\end{restatable}
\begin{proof}
\begin{description}
\item[($1\Leftrightarrow 2$)] Since the open balls are upper sets, these are Scott open iff they are finitely accessible.

\item[($3\To 2$)]
$p$ is Scott continuous when for all $x\in X$ and directed subset $\Delta \subseteq X$ one has 
$
p(x, \bigvee \delta ) = \inf_{d\in \delta}p(x,d)$.
Suppose $p$ is continuous and let $y\in B_\epsilon(x)$. We need to show that there exists $y'\ll y$ such that $y'\in B_\epsilon(x)$. 
This implies that for some $\epsilon'<\epsilon$, $p(y,x)<p(x,x)+\epsilon'$. 
Since $p$ is continuous and $y=\bigvee\{ z\mid z\ll y\}$ we have then 
$
\inf\{ p(z,x)\mid z\ll y \} = p(y,x)<p(x,x)+\epsilon'.
$
This implies in turn that for some $y'\ll y$, $p(y',x)\leq p(x,x)+\epsilon'<p(x,x)+\epsilon$, that is, 
$y'\in B_\epsilon(x)$.

\item[($2\To 3$)]
Suppose that open $p$-balls are finitely accessible, hence Scott open. Let $\Delta\subseteq X$ be a directed set and $x\in X$. We need to prove that $p(x,\bigvee \Delta)=\inf_{d\in\Delta}p(x,d)$. Observe that the "$\leq$" direction directly follows from $d\leq \bigvee \Delta$. To prove the "$\geq$" direction we argue as follows: let $p(x,\bigvee\Delta)=p(x,x)+\delta$, with $\delta \in\BB R_{\geq 0}$. 
Let $\delta'>\delta$, so that we have $\bigvee\Delta\in B_{\delta'}(x)$. Since $B_{\delta'}(x)$ is Scott-open, there exists $w\ll \bigvee\Delta$ such that $w\in B_{\delta'}(x)$. From $w\ll\bigvee\Delta$ it follows that, for some $d\in D$, $w\leq d$ holds, whence $p(d,x)\leq p(w,x)< p(x,x)+\delta'$. 
We have thus proved that for all $ \delta'>\delta$ there exists $ d\in \Delta$ such that $ p(d,x)<p(x,x)+\delta'$, 
which implies then
$
\inf_{d\in \Delta}p(d,x)\leq p(x,x)+\delta= p(x,\bigvee\Delta).$
\end{description}
\end{proof}

To check the converse condition $\C O_\sigma(X)\subseteq \C O_p(X)$, one must show that, given $x\ll y$, one can form open balls around $y$ whose elements all lie way above $x$. This corresponds to showing that the basic open sets
$\twoheaduparrow x=\{y\mid x\ll y\}$ for the Scott topology are metric open.

\begin{remark}\label{rem:leqleqp}
The order $\leq_p$ induced by a PM $p$ on $X$ induces the Scott topology $\mathcal O_{\sigma(p)}(X)$, which need not coincide with $\mathcal O_p(X)$. 
For instance, $p$ might fail to be Scott-continuous (which, by Proposition \ref{prop:continuous}, implies that open $p$-balls are not Scott open): define a variant $q$ of the tree partial metric as 
$q(\alpha,\beta)=\frac{1}{2}p_{\mathrm{tree}}(\alpha,\beta)+\frac{1}{4}$ if $\alpha\neq\beta$ or $\alpha=\beta$ is finite, and as $p_{\mathrm{tree}}(\alpha,\beta)$ if $\alpha=\beta$ is infinite. 
$q$ is still a partial metric and $\leq_q=\leq_p$; yet, letting $\alpha_n$ a directed sequence of finite trees converging to an infinite tree $\alpha$, we have $\lim_n q(\alpha_n, \alpha)=\frac{1}{4}>0= q(\bigvee_n \alpha_n, \alpha)$.
\end{remark}

\begin{remark}[computability of $p(x,y)<\epsilon$]\label{rem:compu}
An immediate and useful consequence of the fact that open balls are Scott open is
that $p(x,y)< \epsilon$ holds precisely when $p(x',y')<\epsilon$ holds for some approximants $x'\ll x$ and $y'\ll y$.
In other words, to verify that $y$ is \emph{close enough} to $x$ it is enough to check that the approximants of $y$ get close enough to the approximants of $x$. When distances between approximants, as well as the relation $b\ll x$ between a point and an approximant, are computable, the property $p(x,y)<\epsilon$ may be (semi-)decidable, even though the exact values $p(x,y)$ are as hard as computing the $\lambda$-theory (usually, $\Pi^0_2$ or worse). 
For instance, in the case of B\"ohm trees, to check that $p_{\mathrm{B\"ohm}}(M,N)<2^{-n}$, it is enough to check that $\mathcal B(M)$ and $\mathcal B(N)$ coincide up to height $n$, a property which can be semi-decided.
\end{remark}

%

%
%

\begin{example}[$\epsilon/\delta$-continuity of contexts]
As $p_{\mathrm{tree}}$ quantifies the Scott topology of trees (cf.~Proposition \ref{prop:ptree}), it quantifies the Scott topology of B\"ohm trees. From the continuity theorem for B\"ohm trees (cf.~\cite{Baren95}) we deduce then the following: for all context $\TT C[-]$ and $\lambda$-term $M$ and for all $\epsilon>0$, there exists $\delta>0$ such that, for all terms $P$,  $p_{\mathrm{B\"ohm}}(P,M)\leq \delta$
implies  $p_{\mathrm{B\"ohm}}(\TT C[P], \TT C[M])\leq \epsilon$.
Another way of stating this is that for all $\TT C[-]$ and $M$, for all $n\in \mathbb N$ there exists $m\in \mathbb N$ such that, if $\mathcal B(P)$ and $\mathcal B(M)$ are the same up to depth $m$, then 
$\mathcal B(\TT C[P])$ and $\mathcal B(\TT C[M])$ are the same up to depth $n$.
\end{example}

%
%
%
%

\subparagraph{Applicative distances and the function space}

The category $\mathsf{Scott}$ of Scott domains and continuous functions is a sub-category of $\mathrm{Top}$ that is, as is well-known, cartesian closed. 
Using Theorem \ref{thm:full} it is possible to define, on each object of $\mathsf{Scott}$, a partial metric that quantifies its topology. However, in common approaches to higher-order languages (e.g.~\cite{Gaboardi2017, Gavazzo2018, Hoshino2023}), one requires distances to be defined in a \emph{compositional} way.

For example, given metric spaces $(X,d_X)$ and $(Y,d_Y)$, a standard way to define a metric on their cartesian product is 
by letting $d_{X\times Y}(\langle x,y\rangle, \langle x',y'\rangle)=d_X(x,x')+d_Y(y,y')$. 
Indeed, a similar construction also works for PMs:

\begin{restatable}{proposition}{cartesian}
Let $X,Y$ be two Scott domains, quantified, respectively, by the partial metrics $p_X,p_Y$. Their cartesian product $X\times Y$ is then quantified by the partial metric $p_{X\times Y}:=\frac{1}{2}(p_X+p_Y)$.
\end{restatable}
\longversion{

\begin{proof}
Given $	a\in X\times Y$, we write $a_X,a_Y$ the elements of $X$ and $Y$ such that $a=(a_X,a_Y)$.
 
	Let us first show that the distance $p_{X\times Y}:=\frac{1}{2}(p_{X}+p_{Y})$ quantifies the order relation on $X\times Y$:
	\begin{itemize}
		
		\item If $a \leq b$, then $a_{X} \leq b_{X}$ and $a_{Y}\leq b_{Y}$, whence $p_{X}(a_{X},b_{X})\leq p_{X}(a_{X},a_{X})$, and $p_{Y}(a_{Y},b_{Y})\leq p_{Y}(a_{Y},a_{Y})$, which implies
		$p_{X\times Y}(a,b)\leq p_{X\times Y}(a,a)$.
		

		\item Suppose $p_{X\times Y}(a,b)\leq p_{X\times Y}(a,a)$; we thus have $p_{X}(a_{X},b_{X}) 
		=2p_{X\times Y}(a,b)-p_Y(a_Y,b_Y)\leq 2p_{X\times Y}(a,a) -p_Y(a_Y,b_Y)
		= p_{X}(a_{X},a_{X}) + p_{Y}(a_{Y},a_{Y}) - p_{Y}(a_{Y},b_{Y}) \leq p_{X}(a_{X},a_{X})$ since $p_{Y}(a_{Y},a_{Y}) - p_{Y}(a_{Y},b_{Y}) \leq 0$ by (P1). Hence $a_{X}\leq b_{X}$.
		
		We can proceed similarly to show $a_{Y}\leq b_{Y}$ and then conclude $a \leq b$.
		
	\end{itemize}
	
	We now show that the distance  quantifies the topology, i.e.~$\C O_{\sigma}(X\times Y)= \C O_{p_{X\times Y}}(X\times Y)$.
	
	\begin{itemize}
		\item $\C O_{\sigma}(X\times Y)\supseteq \C O_{p_{X\times Y}}(X\times Y)$, we show that every metric ball is Scott open:
		\begin{description}
			\item[Upper] Always true.
			
			\item[Finite Access] Suppose again $b \in B_{\epsilon}(a)$; we must show that there exists $b'\ll b$, such that $b' \in B_{\epsilon}(a)$.

			Since $b \in B_{\epsilon}(a)$, we can choose $\delta, \theta > 0$ such that $\epsilon = \delta + \theta$ and $p_{X\times Y}(b,a) < p_{X\times Y}(a,a) + \delta$
			
			Now by induction hypothesis, we can find $c_X \in B_{\theta}(b_X)$ such that $c_x \ll b_x$, same for $Y$. Let $c = (c_X, c_Y)$
			
			$
			\begin{WithArrows}
				p_{X\times Y}(c,a)
				&\leq p_{X\times Y}(c,b) + p_{X\times Y}(b,a) - p_{X\times Y}(b,b) \\
				&= \frac{1}{2}\big (p_{X}(c_{X},b_{X}) + p_{Y}(b_{X},a_{X}) - p_{X}(b_{X},b_{X}) \\
				&+ p_{X}(c_{Y},b_{Y}) + p_{Y}(b_{Y},a_{Y}) - p_{Y}(b_{Y},b_{Y}) \big)\qquad \Arrow{$c_X \in B_{\theta}(b_X)$} \\
				&\leq \frac{1}{2}\big (p_{X}(b_{X},b_{X}) + {\theta} + p_{Y}(b_{X},a_{X}) - p_{X}(b_{X},b_{X}) \\
				&+ p_{X}(b_{Y},b_{Y}) + {\theta} + p_{Y}(b_{Y},a_{Y}) - p_{Y}(b_{Y},b_{Y})\big) \\
				&= p_{X\times Y}(a,b) + \theta \\
				&\leq p_{X\times Y}(a,a) + \delta + \theta \\
				&= p_{X\times Y}(a,a) + \epsilon.
			\end{WithArrows}$
			
		\end{description}
		
		\item $\C O_{\sigma}(X\times Y)\subseteq \C O_{p_{X\times Y}}(X\times Y)$:
		
		We use the fact that the $\twoheaduparrow a$ form a basis for the Scott Topology.
		
		Now, suppose $a \ll b$; we must find $\epsilon >0$ such that $B_{\epsilon}(b)\subseteq \twoheaduparrow a$. Then $\twoheaduparrow a$ will be an open set as the union of all these balls.
		
		From $a \ll b$ it follows that $a_{X} \ll b_{X}$ and $a_{Y} \ll b_{Y}$.
		
		Since the topology coincide for $X$ and $Y$, there are $\epsilon_{X},\epsilon_{Y}>0$ such that $B_{\epsilon_{S}}(b_{S}) \subseteq \twoheaduparrow a_{S}$.
		
		Let $\epsilon = \frac{1}{2}\min\{\epsilon_{X},\epsilon_{Y}\}$, and suppose $c \in B_{\epsilon}(b)$;

		Now 
		$
		\begin{WithArrows}
			p_{X}(c_{X},b_{X}) 
			&=2 p_{X\times Y}(c,b)-p_Y(c_Y,b_Y) \\
			& < 2p_{X\times X}(b,b)-p_Y(c_Y,b_Y)+\min\{\epsilon_{X},\epsilon_{Y}\}\\
			&= p_{X}(b_{X},b_{X}) + p_{Y}(b_{Y},b_{Y}) - p_{Y}(c_{Y},b_{Y}) + \min\{\epsilon_{X},\epsilon_{Y}\} \Arrow{P1} \\
			&\leq p_{X}(b_{X},b_{X}) + \min\{\epsilon_{X},\epsilon_{Y}\}.
		\end{WithArrows}$
		
		So $c_{X} \in B_{\epsilon}(b_X) \subseteq B_{\epsilon_X}(b_X)$, and thus $a_X \ll c_X$. Same for $Y$.
		
		We conclude $a \ll c$.
		
	\end{itemize}
\end{proof}

}

\begin{remark}
In the following discussion we restrict attention to partial metrics valued in $[0,1]$, rather than on $[0,+\infty]$. This is not a limitation, since for any partial metric $p$ with values in $[0,+\infty]$, the partial metric $p^{\leq 1}:X\times X\to [0,1]$ defined by $p^{\leq 1}(x,y) := \frac{p(x ,y)}{1+p(x,y)}$
induces the same topology (cf.~\cite{Myronyk2022}).
\end{remark}

Let us now consider the function space. Given metric spaces $(X,d)$ and $(X',d')$, a standard compositional way to define a metric on the space $\C C(X, X')$ of continuous functions from $X$ to $X'$ is via the $\sup$-condition
$
d_{\sup}(f,g)=\sup_{x\in X}d'(f(x),f(x'))
$. 
Notably, when $X$ is compact, $d_{\sup}$ metrizes the \emph{compact-open} topology on $\C C(X, X')$.
Other compositional metrics on the space of \emph{non-expansive} functions $\mathrm{NExp}(X,X')$, depending on \emph{both} $d$ and $d'$, can be found in the literature \cite{Clementino2006, Honsell2022}. 
Similar compositional definitions are also found in more operational approaches like e.g.~\cite{Reed_2010,Gavazzo2018}.

A common intuition in all these definitions is that two functions are close when their application to close (or even identical) points produces points that are still close. We will call functional distances respecting this idea \emph{applicative distances}.

However, to define an applicative PM on the space of continuous functions, we cannot directly adapt a definition like $d_{\sup}$: unlike for standard (pseudo-)metrics, the $\sup$s of a family of PPMs does \emph{not} define a PPMs. This is due to condition P4, which relies in a \emph{contravariant} way on the medium self-distance $p(z,z)$.

Instead, we will rely on the remark that a continuous function $f:X\to Y$ is uniquely determined by its action on the (countably many) elements of a basis of $X$.
This suggests indeed the definition from the Proposition below:

\begin{restatable}{proposition}{exponential}\label{th:exponential}
Let $X,Y$ be two Scott domains, quantified, respectively, by the PMs $p_X,p_Y$, and let $(a_n)_{n\in \mathbb N}$ be an enumeration of a basis of $X$. Then, for all $0<\theta\leq\frac{1}{2}$, their exponential $X\Rightarrow Y$ is quantified by the PM 
\begin{equation}
p^\theta_{X\To Y}(f,g)= \sum_{n=1}^{\infty}\theta^n p_{Y}(f(a_{n}),g(a_{n})).
\end{equation}
\end{restatable}

In order to prove the result above, we first need to recall a few facts about the domain of Scott-continuous functions. Given  Scott domains $X,Y$, with countable bases $B(X), B(Y)$, the Scott domain 
$\C C(X,Y)$ admits a countable basis formed by all functions of the form $(\twoheaduparrow a \searrow b)$, where $a\in  B(X), b\in B(Y)$, and 
$
(\twoheaduparrow a\searrow b)(x)= b$ in case $a\ll x$, while $
(\twoheaduparrow a\searrow b)(x)= \bot$ otherwise.

Importantly, while $f \ll g$ implies $f(x)\ll g(x)$ for all $x\in X$, the converse does \emph{not} hold.
Rather, the way below relation can be characterized as follows.
\begin{lemma}[cf.~\cite{Escardo1998}]\label{lemma:ll}
For all $f,g\in \C C(X,Y)$, 
$f\ll g$ iff there exists basis elements $a_{1},\dots, a_{n}\in B(X)$ and $b_{1},\dots, b_{n}\in B(Y)$ such that 
\begin{itemize}
\item $b_i\ll g(a_i)$,
\item $\twoheaduparrow a_{i}\ll g^{-1}(\twoheaduparrow b_{i})$\footnote{
Recall that $\C O(X)$ is a continuous domain. For two open sets $U,V\in \C O(X)$, $U\ll V$ holds when any open cover of $V$ has a finite subset which covers $U$. 
}, for all $i=1,\dots,n$,
\item $f \leq \bigvee_{i=1}^{n}(\twoheaduparrow a_{i}\searrow b_{i})$.
\end{itemize}
\end{lemma}
\longversion{
\begin{proof}
We only prove the only if direction, as the other one is easily checked.
Suppose $\twoheaduparrow a \ll g^{-1}(\twoheaduparrow b)$; we will show that 
$\twoheaduparrow a \searrow b \ \ll \ g$.  
Let $H$ be a directed subset of $\C C(X,Y)$ such that $g \leq \bigvee H$. For every $x\in g^{-1}(\twoheaduparrow b)$, we have that $b \ll g(x) \leq \bigvee_{h\in H}h(x)$. Hence there is some $h_{x}\in H$ with $b\ll h_{x}(x)$. Since $x\in h_{x}^{-1}(\twoheaduparrow b)$ and $x$ is arbitrary, we have that $g^{-1}(\twoheaduparrow b)\subseteq \bigcup_{h\in H}h^{-1}(\twoheaduparrow b)$. Since $\twoheaduparrow a\ll g^{-1}(\twoheaduparrow b)$, we conclude that $\twoheaduparrow a\subseteq \bigcup_{h\in I} h^{-1}(\twoheaduparrow b)$, for some finite subset $I\subseteq_{\mathrm{fin}} H$. 
Since $H$ is directed, there exists $h^*\in H$ such that for all $h\in I$, $h\leq h^*$. 
Notice that, if $h\leq h'$, then $h^{-1}(\twoheaduparrow b)\subseteq (h')^{-1}(\twoheaduparrow b)$: if $h(x)\gg b$, then $h'(x)\geq h(x)\gg b$. We deduce then that $\twoheaduparrow a\subseteq \bigcup_{h\in I} h^{-1}(\twoheaduparrow b)\subseteq (h^*)^{-1}(\twoheaduparrow b)$.

Now, for any $x\in X$, if $x\gg  a$, then $(\twoheaduparrow a \searrow b)(x)=b \leq h^*(x)$. Otherwise, $(\twoheaduparrow a\searrow b)(x)=\bot \leq h^*(x)$.
Hence we conclude that $\twoheaduparrow a \searrow b \leq h^*$, and thus that 
$\twoheaduparrow a \searrow b \ \ll \ g$.  

Let $U=\{  \twoheaduparrow a\searrow b \mid \ 
\twoheaduparrow a \ \ll \ g^{-1}(\twoheaduparrow b) \}$ and $U^{*}=\{f_{1}\vee \dots \vee f_{n}\mid n\in \BB N, f_{1},\dots, f_{n}\in U\}$. 
We will show that $g = \bigvee U$: on the one hand, we now that for any $h\in U$, $h\ll g$;
conversely, we have that
\begin{equation}\label{eq:iff}
b \ll g(a)\ \Leftrightarrow \ a \in g^{-1}( \twoheaduparrow b) \ \Leftrightarrow \  \exists a'\in B(X) \ \text{ s.t. }\ 
 a\in \twoheaduparrow a' \ \ll g^{-1}(\twoheaduparrow b) \text{ and } b\ll g(a').
\end{equation}
The last equivalence is proved as follows: if $a \in g^{-1}( \twoheaduparrow b)$, then $U=g^{-1}(\twoheaduparrow b)$ is an open set containing $a$, and thus there exists $a' \ll a $ (which we can suppose to be a basis element) such that $a'\in U$, whence $a\in \twoheaduparrow a' \subseteq U$ and $b\ll g(a')$. The converse direction is immediate.
 Now we have that 
 \begin{align*}
( \bigvee U)(a)& = \bigvee \{ b \mid \exists a'. \  a\in \  \twoheaduparrow a' \ \ll \ g^{-1}(\twoheaduparrow b)\text{ and }b\ll g(a') \} = \bigvee\{ b \mid  b  \ll  g(a) \} =g(a),
 \end{align*} 
 from which we deduce $g= \bigvee U$.
 
Since for any $h\in U$, $h\leq g$, we also have that for any $h\in U^{*}$, $h\leq g$, and from $g= \bigvee U$ we can also deduce $g= \bigvee U^{*}$.
Now, since $U^{*}$ is directed, from $f\ll g$ and $g\leq \bigvee U^{*}$, we deduce that there exists $h_1,\dots, h_n\in U^*$ such that $f\leq \bigvee_i h_i$, which concludes our proof.
\end{proof}
}

We now have all ingredients to prove Proposition \ref{th:exponential}
\begin{proof}[Proof of Proposition \ref{th:exponential}]
	Let us first show that the distance $p^\lambda_{X\To Y}$ quantifies the order relation of $X \To Y$.
	\begin{itemize}
		\item If $f\leq g$, then for all $a\in X$, $f(a)\leq g(a)$, whence $p_{Y}(f(a),g(a))\leq p_{Y}(f(a),f(a))$.
		This implies then 
			$p^\lambda_{X\To Y}(f,g)=\sum_{i=1}^{\infty} \lambda^np_{Y}(f(a_{i}),g(a_{i})) 
			 \leq \sum_{i=1}^{\infty}\lambda^np_{Y}(f(a_{i}),f(a_{i})) 
			= p^\lambda_{X\To Y}(f,f)$.
		
		\item Suppose $p^\lambda_{X\To Y}(f,g)\leq p^\lambda_{X\To Y}(f,f)$;
		if $f(a) > g(a)$ holds for some $a\in X$, then since $f$ is continuous $f(a_{i})>g(a_{i})$ must hold for some element $a_{i}$ of the basis $\C B(X)$. Now, since for all $a\in X$, we have $p_{Y}(f(a),g(a))\geq p_{Y}(f(a),f(a))$, we deduce 
		$p^\lambda_{X\To Y}(f,g)=\sum_{i=1}\lambda^n p_{Y}(f(a),g(a))
		> \sum_{i=1}\lambda^n p_{Y}(f(a),f(a))=p^\lambda_{X\To Y}(f,f),$
		against the hypothesis. We conclude that $f\leq g$.

	\end{itemize}
	
	Let us now show that the distance quantifies the topology. 
	
	\begin{description}
		\item[$\C O_{\sigma}(X\To Y)\supseteq \C O_{p_{X\To Y}}(X\To Y)$:]
		
		We have to show that open p-balls are Scott-open.
		
		\begin{description}
			\item[Upper] Always true.

			\item[Finite Access]
			
			We now show that for all $g\in B_{\epsilon}(f)$ we can find some $h\ll g$ such that $h\in B_{\epsilon}(f)$. Let then $g\in B_{\epsilon}(f)$, so that $p^\lambda_{X\To Y}(f,g)<p^\lambda_{X\To Y}(f,f) +\epsilon$.
			Observe that this implies that we can find positive reals $\theta, \delta>0$ such that $\theta+\delta\leq\epsilon$ and 
			$p^\lambda_{X\To Y}(f,g)< p^\lambda_{X\To Y}(f,f) +\delta$.
			Let $N$ be such that $\sum_{n > N}^\infty \lambda^n \leq \frac{\theta}{2}$.
			For all $n\leq N$, fix some $b_{n}\in B_{\frac{\theta}{2} }(g(a_{n}))$ such that $b_{n}\ll g(a_{n})$, and some basis element $c_{n}\ll a_{n}$. 

			Let now $h= \bigvee_{i=1}^{N}(\twoheaduparrow c_{i} \searrow b_{i})$. 
			From $b_{i}\ll g(a_{i})$ it follows that $a_{i}\in g^{-1}(\twoheaduparrow b_{i})$, and thus that 
			$\twoheaduparrow c_{i}\ll  g^{-1}(\twoheaduparrow b_{i})$. 
			By Lemma \ref{lemma:ll} this implies that $h \ll g$. 
			
			Let us show that $h\in B_{\epsilon}(f)$. For all $n< N$, we have 
			$
			p_{Y}(h(a_{n}),g(a_{n}))\leq p_{Y}(b_{n},g(a_{n})) < p_{Y}(g(a_{n}),g(a_{n}))+\frac{\theta}{2}
			$, 
			whence 
			$p_{Y}(h(a_{n}),g(a_{n}))-p_{Y}(g(a_{n}),g(a_{n}))<\frac{\theta}{2} \qquad (n\leq N).$
			The following computation show then that $h\in B_{\epsilon}(f)$:
			{\footnotesize
				\begin{align*}
					p^\lambda_{X\To Y}(h,f) &= \sum_{n=1}^{\infty}\lambda^n p_{Y}(h(a_{n}),f(a_{n})) \\
					&\leq  \sum_{n=1}^{\infty}\lambda^n \Big (
					p_{Y}(h(a_{n}),g(a_{n}))
					+
					p_{Y}(g(a_{n}),f(a_{n}))
					-
					p_{Y}(g(a_{n}),g(a_{n}))
					\Big )\\
					&\leq 
					\sum_{n=1}^{N}\lambda^n \Big (p_{Y}(h(a_{n}),g(a_{n}))-
					p_{Y}(g(a_{n}),g(a_{n}))\Big) \\
					& \qquad\qquad +
					\sum_{n>N}^{\infty}\lambda^n p_{Y}(h(a_{n}),g(a_{n}))
					+
					p^\lambda_{X\To Y}(g,f)\\
					&< 
					\sum_{n=1}^{N}\lambda^n\frac{\theta}{2}
					+ \sum_{n>N}^\infty\lambda^n +p^\lambda_{X\To Y}(f,f)+ \delta\\
					&\leq 
					\frac{\theta}{2}
					+ \frac{\theta}{2} +p^\lambda_{X\To Y}(f,f)+ \delta
					\leq
					p^\lambda_{X\To Y}(f,f)+\epsilon.
				\end{align*}
			}

		\end{description}
		
		\item[$\C O_{\sigma}(X\To Y)\subseteq \C O_{p_{X\To Y}}(X\To Y)$:]

		It suffices to show that the basic Scott open sets $\twoheaduparrow f$ contain an open p-ball $B_\epsilon(g)$ around any of its points $g\in\ \twoheaduparrow f$. So, suppose $f \ll g$:
		by Lemma \ref{lemma:ll} there exists $c_{1},\dots, c_{n}\in X$, $b_{1},\dots,b_{n}\in Y$ such that 
		$b_i\ll g(c_i)$, $\twoheaduparrow c_{i} \ll g^{-1}(\twoheaduparrow b_{i})$ and $f\leq \bigvee_{i}(\twoheaduparrow c_{i}\searrow b_{i})$. 
		From $b_{i}\ll g(c_{i})$ it follows that there exists $\epsilon_{i}>0$ such that
		$B_{\epsilon_{i}}(g(c_{i}))\subseteq \  \twoheaduparrow b_{i}$. 
		Let $N$ be such that for all $i=1,\dots, n$, $c_{i}$ has an index $\leq N$ in the enumeration $a_n$ of $\C B(X)$. Let $\epsilon= \lambda^{N}\min\{\epsilon_{1},\dots, \epsilon_{n}\}$. 
		
		We claim that $B_{\epsilon}(g)\subseteq \ \twoheaduparrow f$: let $h\in B_{\epsilon}(g)$, 
		then, from 
		$\sum_{n}\lambda^n \big(p_{Y}(h(a_{n}),g(a_{n}))- p_{Y}(g(a_{n}), g(a_{n}))\big )< 
		\epsilon$, we deduce, for $i=1,\dots,n$, $
			p_{Y}(h(c_{i}),g(c_{i})) < p_{Y}(g(c_{i}),g(c_{i}))+\lambda^{-i}\epsilon
			\leq p_{Y}(g(c_{i}),g(c_{i}))+\epsilon_{i}$, that is, $h(c_{i})\in B_{\epsilon_{i}}(g(c_{i}))$.
		We deduce that $b_{i} \ll h(c_{i})$, and thus that $f\leq \bigvee_i(\twoheaduparrow c_{i} \searrow b_{i} ) \ll h$.
		We can thus conclude that $f \ll h$.
	\end{description}
\end{proof}

The distances $p^\theta_{X\To Y}(f,g)$ are defined by infinite series, which are convergent by our assumption that $p_X,p_Y$ are bounded by 1. However, for any $\epsilon>0$, the verification that 
$p^\theta_{X\To Y}(f,g)<\epsilon$ can be reduced to a finitary test:
\begin{restatable}{lemma}{finite}\label{lemma:finite1}
For all continuous functions $f,g:X\to Y$, for all $n>0$, there exists $N\in \mathcal O(n)$ such that, if 
$p_Y(f(a_i),g(a_i))<2^{-(n+1)}$ holds for all $i=1,\dots,N$, then 
$p^\theta_{X\To Y}(f,g)<2^{-n}$.
\end{restatable}
\begin{proof}
Let $\epsilon=2^{-n}$. We must choose $N$ so that $\sum_{i>N}^\infty\lambda^i <\frac{\epsilon}{2}$. 
Since $\sum_{n=1}^\infty \theta^n\leq 1$, this corresponds to requiring 
$\sum_{n=1}^N \theta^n  > 1-\frac{\epsilon}{2}$, or, equivalently, $\frac{1-\theta^{N+1}}{1-\theta}- 1 > 1-\frac{\epsilon}{2}.$ 
A few computations yield then the condition
$
N +1 > \log\left( 3\theta +\theta^2\epsilon+\epsilon \right)\in \C O(\log \epsilon)
$. 

Let us show that under this condition the claim holds.
Suppose $p^\theta(f(a_i),g(a_i))\leq \frac{\epsilon}{2}$ holds for $i=1,\dots,N$. Then we have
\begin{align*}
p^\theta(f,g)&= \sum_{k=1}^\infty\theta^n q(f(a_k), g(a_k))\\
&=\left( \sum_{k=1}^N\theta^n q(f(a_k), g(a_k))\right )+
\left( \sum_{k>N}^\infty\theta^n q(f(a_k), g(a_k))\right )\\
&\leq\left( \sum_{k=1}^N\theta^n \frac{\epsilon}{2}\right )+
\frac{\epsilon}{2}\leq \epsilon.
\end{align*}
%
%
\end{proof}

The intuition behind the test above is that for $N$ high enough, the infinite sum $\sum_{n\geq N}^\infty\theta^n$ gets too small to actually matter in checking that $p^\theta_{X\To Y}(f,g)$ is smaller than $\epsilon$, and one is thus reduced to the \emph{finite} sum $\sum_{n=1}^N\theta^np_{Y}(f(a_{n}),g(a_{n}))$.
This is indeed a key ingredient in showing that open balls of functions are finitely accessible, and in particular, that if $g\in B_\epsilon(f)$, this only depends on {finitely} many values of $g$.

Conversely, from $p^\theta_{X\To Y}(f,g)<\epsilon$, one can deduce bounds $p_Y(f(a_{n}),g(a_{n}))<\theta^{-n}\epsilon$ for \emph{all} $n\in \BB N$, although such bounds become more and more loose as $n$ increases, due to the exponential scaling factor $\theta^{-n}$.

%

\begin{example}[RealPCF]\label{ex:realpcf}
The language RealPCF \cite{Escardo1996} is an extension of PCF with a type $\B I$ for \emph{partial real numbers} (i.e.~finite approximations of real numbers or, equivalently, computable closed intervals) and primitives for computable analysis, with a canonical Scott semantics in which $\B I$ is interpreted via the domain $\intv$. This is perfectly in line with the quantification of $\intv$ we presented in Example \ref{interval}, which sees smaller and smaller intervals as providing more and more information. Via the applicative distances just presented, we obtain then a quantification of the Scott semantics of full RealPCF.

%
%

\end{example}

\begin{example}[Scott topologies of open and closed sets]
Given a topological space $X$, one can endow the space $\mathcal O(X)$ of its open sets with the Scott topology induced by the inclusion order, as well as the (homeomorphic) space $\mathcal C(X)$ of its closed sets under the Scott topology induced by the reversed inclusion order. 

Whenever $X$ is exponentiable in $\mathrm{Top}$ (which is the case, in particular, whenever $X$ is a Scott domain), the bijection $h:\mathcal O(X)\simeq \mathrm{Top}(X,S)$, where $S$ is the Sierpinski space and $h(U)$ is the characteristic function of $U$, is a homeomorphism \cite{Escardo2001}.
Given a countable basis $(x_n)_n$ of $X$, and weights $\theta_n$ with $\sum_n\theta_n\leq 1$, 
we can then quantify $\mathcal O(X)$ and $\mathcal C(X)$ via
\begin{align*}
p^{\mathcal O}_{x_n, \theta_n}(U,V)&=\sum_{n=1}^\infty \theta_n\cdot  s(h(U)(x_n),h(V)(x_n))=
\sum\left\{ \theta_n  \  \vert \
x_n\notin U \lor x_n\notin V
\right\},\\
p^{\mathcal C}_{x_n, \theta_n}(C,D)&=p^{\mathcal O}(\overline C,\overline D)=
\sum\left\{ \theta_n  \  \vert \
x_n\in C \lor x_n\in D
\right\}.
\end{align*}
\end{example}

\begin{example}[B\"ohm trees as closed sets]\label{ex:altbohm}
Consider the poset $\C A$ of partial terms. Let $\mathrm{Ide}(\C A)$ be the dcpo of \emph{ideals} of $\C A$, that is, of lower directed subsets of $\C A$. 
Observe that any B\"ohm tree $\C B(M)\subseteq \C A$ is an element of $\mathrm{Ide}(\C A)$, and the set 
$\downarrow \C B(M)=\{U \mid U\subseteq \C B(M)\}\subseteq \mathrm{Ide}(\C A)$ is a closed set under the Scott topology of $\mathrm{Ide}(\C A)$.
%
%
Given an enumeration $A_n$ of partial terms and weights $\theta_n$, we can then define an alternative $\lambda$-PPM by letting $p^{\mathcal B}_{A_n, \theta_n}(M,N)= p^{\C C}_{\downarrow A_n, \theta_n}(\downarrow\C B(M), \downarrow \C B(N))=
\sum_n \{\theta_n\  \vert \  A_n\not \leq M \ \text{or} \ A_n \not \leq N  \}$. While they produce different distances, $p^{\mathcal B}_{A_n, \theta_n}$ and $p_{\mathrm{B\"ohm}}$ quantify the same topology, i.e.~$p_{\mathrm{B\"ohm}} \sqsupseteq \sqsubseteq p^{\mathcal B}$.

\end{example}

\begin{example}[Scott topology of the power set]\label{ex:altpower}
A countable set $X$ is (trivially) a domain for the order given by equality, and its Scott topology coincides with the indiscrete topology, i.e.~$\mathcal O(X)=\mathcal P(X)$. Given an enumeration $x_n$ of $X$, the Scott topology on $\mathcal P(X)$ is thus quantified by 
$
p^{\mathcal P}_{x_n, \theta_n}(A,B)=\sum\left\{\theta_n \  \vert \
x_n\notin A \lor x_n\notin B
\right\},
$
for $A,B\subseteq X$.
\end{example}

\section{Quantifying a Reflexive Object}\label{sec5}

The denotational models for the untyped $\lambda$-calculus correspond to the \emph{reflexive objects} within some cartesian closed category, that is, the objects $X$ satisfying the isomorphism $X\simeq X\to X$.
Within cpo-enriched categories, reflexive objects can be obtained by a direct limit construction, whose paradigmatic example is Scott's $D_\infty$ model.
In this section we show how to quantify this model via applicative distances, at the same time illustrating a technique that could be adapted to other similar constructions, like e.g.~the reflexive object within the relational model \cite{Manzo2007}.

\subparagraph{Quantifying Scott's $D_\infty$}

Let us recall the idea of the direct limit construction of a reflexive object.  
\longversion{A more detailed recall of the construction of the model $D_\infty$ is in the Appendix.}  
 One starts from some bounded complete domain $D$, and constructs a sequence of spaces $D_0:=D$, $D_{n+1}:D_n\to D_n$, together with maps $i_n:D_n\to D_{n+1}$ and $j_n:D_{n+1}\to D_n$ forming a pair $(i_n,j_n)$ called an \emph{injection/retraction pair}, that is, satisfying $j_n\circ i_n=\mathrm{id}_{D_{n+1}}$ and $i_n\circ j_n\leq \mathrm{id}_{D_n}$. 
One obtains then a reflexive object $D_\infty\simeq D_\infty\to D_\infty$ as the direct limit of the sequence $D_n \stackrel{i_n}{\to}D_{n+1}$, as well as injection-retraction pairs $i_{n\infty}:D_n\to D_\infty$, $j_{\infty n}:D_\infty\to D_n$.

Notice that an element $x$ of $D_\infty$ yields, for any $n$, a function $x_n:=j_{\infty n}(x)\in D_n= D_{n-1}\to D_{n-2}\to \dots \to D_0$; conversely, any compact element $x\in D_n$ yields a compact element $i_{n\infty}(x)\in D_\infty$, and such elements form indeed a basis of $D_\infty$.

Suppose now that the starting space $D$ is quantified by some PM $p$. 
Using the applicative metrics from the previous section we can quantify all the $D_n$ by letting $p_0:=p$ and 
$p_{n+1}(x,y)=\sum_{i=1}^\infty\frac{1}{2^i}p_n(x(a^n_i),y(a^n_i))$, where $(a^n_i)_i$ is an enumeration of the basis elements of $D_n$. 
We obtain then a PM quantifying $D_\infty$ by letting 
\begin{equation*}
\begin{split}
p_\infty(x,y)&=\sum_{n=1}^\infty\frac{1}{2^n}p_n(x_n,y_n)\\
&=
\sum_{n,k_{n-1},\dots, k_0=1}^\infty\left(\frac{1}{2^{n+k_{n-1}+\dots +k_0}}\right)\cdot p\left (x_n(a^{n-1}_{k_{n-1}})\dots (a^0_{k_0}),y_n(a^{n-1}_{k_{n-1}})\dots (a^0_{k_0})\right ).
\end{split}
\end{equation*}

Intuitively, the distance $p_\infty$ compares $x$ and $y$ by considering all possible ways of evaluating the functions $x_n,y_n\in D_n$ on $n$ basis elements of the corresponding spaces $D_{n-1},\dots, D_0$. 
As for the applicative metrics from the previous section, while the distances $p_\infty(x,y)$ are defined via infinite series, one can check that $x\in B_\epsilon^{p_\infty}(y)$ by a finitary criterion.

\begin{restatable}{lemma}{finitaryscott}
For all $x,y\in D_\infty$ and $n>0$, there exists $N\in \mathcal O(n)$ such that, if, for all $i, k_0,\dots, k_{i-1}\leq N$, 
$p (x_i a_{i-1}^{k_{i-1}}\dots a_0^{k_0}, y_i a_{i-1}^{k_{i-1}}\dots a_0^{k_0})<2^{-(n+1)}$, then 
$p_{\infty}(x,y)<2^{-n}$.
\end{restatable}
\begin{proof}
We must find $N$ satisfying $\sum_{n,k_{n-1},\dots, k_0>N}^\infty \frac{1}{2^{n+k_{n-1}+\dots+k_0}}  < \frac{\epsilon}{2}$.
Notice that if $N$ satisfies $\sum_{n>N}^\infty \frac{1}{2^{n}}  < \frac{\epsilon}{2}$, then it also satisfies the other condition, so we can argue as for
Lemma \ref{lemma:finite1}.
\end{proof}

\shortversion{The following result is proved in detail in the long version.}

\begin{restatable}{theorem}{metricscott}\label{th:dscott}
The partial metric $p_\infty$ quantifies the Scott topology of $D_\infty$.
\end{restatable}
\begin{proof}
We will exploit a few properties of the maps $i_{nm}$, proved in the Appendix:
\begin{enumerate}
\item[i.] 	For all $n\in \BB N$, $x\in X_n$ and $y\in X_\infty$, 
	$x \ll y_n$ $\To$ $i_{n\infty}(x)\ll y$.
	
\item[ii.] 	For all $x,y\in X_\infty$, $x\ll y$ iff there exists $N\in \BB N$ and 
	$w_1,\dots,w_k\in X_{N}$ such that
	$
	w_1,\dots, w_k \ll y_N
	$
	 and
	$
	x \leq  i_{N\infty}(w_1\vee\dots\vee w_k)
	$.

\item[iii.] 	For all $n\in \BB N$, $x\in X_n$ and $y\in X_\infty$, 
	$i_{n\infty}(x)\ll y$ $\To$ $\exists N\forall k\geq N, i_{n(n+k)}(x)\ll y_{n+k}$.

\end{enumerate} 

	\begin{description}
		\item[$\C O_\sigma(D_\infty)\supseteq\C O_{p_\infty}(D_\infty)$:]

		Let $y\in B_\epsilon(x)$. We need to find $y'\in D_\infty$ such that $y'\in B_\epsilon(x)$ and 
		$y'\ll y$. 
		From $p_{\infty}(y,x)<p_{\infty}(x,x)+\epsilon$ it follows that we can find $\theta, \delta>0$ such that 
		$p_{\infty}(y,x)<p_{\infty}(x,x)+\delta$ and $\delta+\theta\leq \epsilon$.
		Let $N$ be such that $\sum_{n>N}^\infty\frac{1}{2^n}< \frac{\theta}{2}$. 
		Since the $p_n$-balls are Scott open, for all $n\leq N$, we can find some $z_n\in B_{\frac{\theta}{2}}(y_n)$ 		such that 
		$z_n\ll y_n$.
		Observe that by (i.) we have $i_{n\infty}(z_n)\ll y$.
		This implies in particular that the join $\bigvee_{n=1}^N i_{n\infty}(z_n)$ exists in $D_\infty$.
		Define $y':= \bigvee_{n=1}^N i_{n\infty}(z_n)$. Notice that $y'\ll y$ holds so we just have to check that 
		$y'\in B_\epsilon(x)$.

		First recall that, by antimonotonicity of $p_n$, $p_n(a\vee a', b)\leq \min\{p_n(a,b), p_n(a',b)\}$.
		Now, for all $n\leq N$, we have that $y'_n=j_{\infty n}(y')=\left( \bigvee_{k<N}i_{k,n}(z_k)\right ) \vee z_n \vee
		\left( \bigvee_{n<k\leq N}j_{k,n}(z_k)\right)$. Then we deduce $p_n(y'_n,y_n)\leq p_n(z_n, y_n)< p_n(y_n,y_n)+ \frac{\theta}{2}$.
		We can now compute
		{\footnotesize
		\begin{align*}
			p_\infty(y',x)&= \sum_{n=1}^\infty\frac{1}{2^n}p_n(y'_n, x_n) \\
			&\leq
			\sum_{n=1}^\infty\frac{1}{2^n}\Big( p_n(y'_n, y_n)+p_n(y_n, x_n)-p(y_n,y_n) \Big )\\
			&\leq
			\left(\sum_{n=1}^\infty\frac{1}{2^n} p_n(y'_n, y_n)-p_n(y_n,y_n)  \right) +p_\infty(y,x)\\
			&= \left(\sum_{n=1}^N\frac{1}{2^n} (p_n(y'_n, y_n)-p_n(y_n,y_n))\right) \\ & \qquad\qquad +\left( \sum_{n>N}^\infty\frac{1}{2^n} (p_n(y'_n, y_n)-p_n(y_n,y_n)) \right)+p_\infty(y,x)\\
			&< \left(\sum_{n=1}^N\frac{1}{2^n}\frac{\theta}{2}\right) + \frac{\theta}{2} +p_\infty(x,x) +\delta
			\leq p_\infty(x,x)+\theta+\delta
		\leq p_\infty(x,x)+\epsilon.
		\end{align*}
		}
		
		%
		%

		\item[$\C O_\sigma(X_\infty)\subseteq\C O_{p_\infty}(X_\infty)$]
		
		Suppose $ x\ll y$. We need to find $\epsilon>0$ such that $B_\epsilon(y)\subseteq\ \twoheaduparrow x$.
		By (ii.) there exists $N$ and 
		$w_1, \dots, w_k\in X_{N}$ such that
		$
		x \leq  i_{N\infty}(w_1\vee\dots\vee w_k)\ll y
		$.
		By (iii.) there exists $N'\geq N$ such that
		$
		i_{NN'}(w_j) \ll y_{N'}
		$. 
		Observe that $i_{N'\infty}(i_{NN'}(u))= i_{N\infty}(u)$, which implies that 
		$x \leq \bigvee_j  i_{N'\infty}(i_{NN'}(w_j) )  $.

		For each $j=1,\dots, k$ we can find then $\epsilon_j>0$ such that 
		$B_{\epsilon_j}(y_{N'})\subseteq \ \twoheaduparrow i_{NN'}(w_j)$.
		Let $\epsilon:= 2^{-(N'+1)}\min\{\epsilon_j\mid j=1,\dots, k\}$.
		Suppose $z\in B_\epsilon(y)$: for all $j=1,\dots,k$, from $p_\infty(z,y)\leq \epsilon$ we deduce
		$p_{N'}(z_{N'}, y_{N'})\leq 2^{N'}\epsilon < \epsilon_j$, whence 
		$z_{N'}\in B_{\epsilon_j}(y_{N'})$, which forces
		$i_{NN'}(w_j) \ll z_{N'}$. 
		By (i.) the last inequality implies 
		$i_{N'\infty}(w_j)=i_{N'\infty}(i_{NN'}(w_j) )\ll z$, and we thus obtain
		$x \leq  \bigvee_j  i_{N'\infty}(i_{NN')}(w_j) ) \ll z$, that is, $x\ll z$.
	\end{description}
\end{proof}

\subparagraph{The Scott $\lambda$-PPM }

The interpretation of closed $\lambda$-terms in the Scott model $D_\infty$, for $D$ an arbitrary algebraic domain quantified by a PM $p$, yields a PPM $p_{\mathrm{Scott}}(M,N):=p_\infty(\model M, \model N)$, where $\model M\in D_\infty$ indicates the interpretation of $M$ inside $D_\infty$. When $D$ is non-trivial (i.e.~$D\neq \{\bot\}$), using well-known properties of the Scott model, $p_{\mathrm{Scott}}(M,N)$ yields an extensional and sensible $\lambda$-PPM.

The result below relates
 $p_{\mathrm{Scott}}$ to the other $\lambda$-PPMs discussed in Section 3.

\begin{restatable}{proposition}{metscott} 
$p_{\mathrm{B\"ohm}}\sqsubset p_{\mathrm{Scott}}\sqsubset p_{\mathrm{ctx}}$.
\end{restatable}
\begin{proof}
\begin{description}
\item[($p_{\mathrm{B\"ohm}}\sqsubset p_{\mathrm{Scott}}$)]
We exploit the \emph{approximation theorem} for $D_\infty$ \cite{Baren95} which says that, for any closed $\lambda$-term $M$, letting $\Lambda_\bot^o$ be the set of closed partial terms and 
 $\model{-}:\Lambda^o_\bot\to D_\infty$ the interpretation function, $\model{M}=\bigvee\{ \model{A}\mid A\leq M\}$.

Consider the open ball $B_\epsilon^{p_{\mathrm{Scott}}}(M)$. Fix $p, m_1,\dots, m_p\in\BB N$ and $\delta_{i,j_1,\dots, j_i}>0$, for $i\leq p$ and $j_\ell\leq m_\ell$ such that 
$\sum_{i> p, j_\ell > m_\ell}^\infty \frac{1}{2^{i +j_1 +\dots+j_i}}\leq\frac{\epsilon}{2}$ and $\sum_{i=1}^K\sum_{j_\ell=1}^{m_\ell} \frac{\delta_{i,j_1,\dots, j_i}}{2^{i +j_1 +\dots+j_i}}<\frac{\epsilon}{2}$. 

Call $\vec a_{i,j_1,\dots, j_i}$ the sequence of basis elements corresponding to the indices $i,j_1,\dots, j_i$, so that 
$p_\infty(x,y)=\sum_{i,j_1,\dots, j_i}\frac{1}{2^{i +j_1 +\dots+j_i}}p_0(x(\vec a_{i,j_1,\dots, j_i}),y(\vec a_{i,j_1,\dots, j_i}))$. 

Then we have that, for any term $P$, if for all $i\leq p$ and $j_\ell\leq m_\ell$, 
$\model{P}(\vec a_{i,j_1,\dots, j_i})\in B_{\delta_{i,j_1,\dots, j_i}}^{p_0}(\model{M}(\vec a_{i,j_1,\dots, j_i}))$, then 
$p_\mathrm{Scott}(P,M)=p_\infty(\model P, \model M)< p_\mathrm{Scott}(P,M)+\epsilon$, whence 
$P\in B_\epsilon^{p_{\mathrm{Scott}}}(M)$.

Since $D$ is algebraic, any open ball $B_{\delta_{i,j_1,\dots, j_i}}^{p_0}(\model{M}(\vec a_{i,j_1,\dots, j_i}))$ contains some compact element $c_{i,j_1,\dots, j_i}$. From $c_{i,j_1,\dots, j_i} \ll \model{M}(\vec a_{i,j_1,\dots, j_i})$, by the approximation theorem and the fact that $c_{i,j_1,\dots, j_i}$ is compact, we deduce then that $c_{i,j_1,\dots, j_i} \ll \model{A_{i,j_1,\dots, j_i}}(\vec a_{i,j_1,\dots, j_i})$ must hold for some partial term $A_{i,j_1,\dots, j_i}\leq M$.

Let now $A=\bigvee_{i\leq p,j_\ell\leq m_\ell}A_{i,j_1,\dots, j_i}$. Notice that $A$ is a partial term, and that $A\leq M$. 
We claim that $A\in B_{\delta_{i,j_1,\dots, j_i}}^{p_0}(\model{M}(\vec a_{i,j_1,\dots, j_i}))$: for any  $i,j_1,\dots, j_i$, 
we have 
$\model{A}(\vec a_{i,j_1,\dots, j_i})\geq  \model{A_{i,j_1,\dots, j_i}}(\vec a_{i,j_1,\dots, j_i})\in B_{\delta_{i,j_1,\dots, j_i}}^{p_0}(\model{M}(\vec a_{i,j_1,\dots, j_i}))$, and we can conclude since open balls are upper closed.

We can thus conclude that $A\in B^{p_{\mathrm{Scott}}}_\epsilon(M)$. Let $n$ be the height of $A$ and $\theta=2^{-n}$; then 
for all $P$, if $P\in B_\theta^{p_{\mathrm{B\"ohm}}}(M)$,
it must hold $\C B(P)\geq A$, and thus $A\leq P$, which implies then 
$P\in B^{p_{\mathrm{Scott}}}_\epsilon(M)$.

The strictness follows from the fact that the associated $\lambda$-theories $\C B$ and $\C H^*$ are strictly included, 
as argued at the end of Remark \ref{rem:ctxbohm} for the case of $p_{\mathrm{ctx}}$.

\item[($ p_{\mathrm{Scott}}\sqsubset p_{\mathrm{ctx}}$)]
By Proposition \ref{prop:coarsest}, we only need to prove strictness. 
Recall that the basic open sets for the Scott topology are of the form $\twoheaduparrow b=\{x\mid b\ll x\}$, where $b$ is an element of the basis. 

Fix a compact element $c\neq \bot$ of $D_0$.
Let $c_1:= \twoheaduparrow c\searrow c$, a compact basis element of $D_{1}$,
$c_2:= \twoheaduparrow c_1\searrow c_1$, a compact basis element of $D_{2}$, such that $c_2(c_1)=c_1$ (using the fact that $c_1\ll c_1$). More generally, define $c_{n+1}:=\twoheaduparrow c_n\searrow c_n$ a compact basis element of $D_{n+1}$ such that $c_{n+1}(c_n)(c_{n-1})\dots (c_1)=c_1$. Observe that for all $n$, $c_n\ll \mathrm{id}_{D_n}\ll \model{I}$.

For all $k\geq 1$, let 
$P_k=\lambda y_1.\dots \lambda y_k. \lambda x.x$.
We claim that $c_{n}\not\ll \model{P_{n+k}}$ holds for all $n,k>0$.
To prove this we will define terms $e_n\in D_n$ such that 
$\model{P_{n+k}}(e_{k+n-1})\dots (e_0)= \bot$ while 
$i_{n(k+n)}(c_{n})(e_{k+n-1})\dots (e_0)\neq \bot$.

%

Define $d_0:= c$ and $d_{n+1}:=\lambda x.d_n$. Let us show, by induction on $n$, that $i_n(d_n)=d_{n+1}$ and 
$j_n(d_{n+1})=d_n$: 
we have $i_0(d_0)=\lambda x.c=d_1$ and $j_0(d_1)=d_1(\bot)=c=d_0$; for the inductive case we have
$i_{n+1}(d_{n+1})=\lambda x. i_n(d_{n+1}(j_n(x)))=\lambda x.i_n(d_n)=\lambda x.d_{n+1}=d_{n+2}$, and 
$j_{n+1}(d_{n+2})= \lambda x. j_n(d_{n+2}(i_n(x)))=\lambda x. j_n(d_{n+1})=\lambda x.d_n=d_{n+1}$. 
So, more generally, we have $i_{nm}(d_n)=d_m$ and $j_{mn}(d_m)=d_n$. 
Hence, for all $n$ and $k>0$, we have 
\begin{align*}
\Big(i_{n(k+n)}(c_{n})\Big)(d_{k+n-1})&(d_{k+n})\dots (d_{k})(\bot)\dots (\bot)\\
&= 
i_{0k}\Big(c_n(j_{(k+n-1)(n-1)}(d_{k+n-1}))\dots (j_{k0}(d_k))\Big)(\bot)\dots (\bot) \\
&= i_{0k}\Big( c_n(d_{n-1})\dots (d_0)\Big)(\bot)\dots (\bot) \\
&=i_{0k}(d_0)(\bot)\dots (\bot) \\
&= d_{k}(\bot)\dots (\bot) =c.
\end{align*}

Given $n,k>0$, let $(e_{n+k-1})\dots (e_0)$ stand for $(d_{k+n-1})\dots (d_{k})(\bot)\dots (\bot)$, with $(\bot)$ repeated $k$ times. We have, on the one hand, 
$
i_{n(k+n)}(c_{n})(\vec e) =  c$, while, on the other hand,   
$\model{P_{k+n}}(\vec e)= \bot$. This proves that $c_{n}\not\ll P_{n+k}$.

%
%
%
%
%
%

%
%

%
%
%

Now, since $\twoheaduparrow {c_2}$ is open and $I \in\  \twoheaduparrow c_2$, there exists $\epsilon>0$ such that for al $Q\in B_\epsilon^{\mathrm{Scott}}(I)$, $\model Q\in\  \twoheaduparrow c_2$.  
This implies then that for all $k>0$, $P_{k+2}\notin B_\epsilon^{\mathrm{Scott}}(I)$, since $c_2\not \ll P_{k+2}$.

We will now show that \emph{any} open $p_{\mathrm{ctx}}$-ball around $I$ contains infinitely many terms $P_k$.
This implies that $B_\epsilon^{\mathrm{Scott}}(I)$, which does not contain $P_{k}$ as soon as $k>2$, contains no open $p_{\mathrm{ctx}}$-ball around $I$, as desired.

Let $\delta>0$ and consider the ball $B_\delta^{\mathrm{ctx}}(I)$.
Then there exists finitely many contexts $\TT C_1,\dots, \TT C_m$ such that $\TT C_i[I]$ is solvable for all $i=1,\dots,m$ and for all terms $Q$, if $\TT C_i[Q]$ is solvable for all $i=1,\dots, m$, then
 $Q\in B_\delta^{\mathrm{ctx}}(I)$.

Let $\TT C_i[-]$ be a context such that $\TT C_i[I]$ is solvable. 
Observe that in the reduction of $\TT C_i[I]$, either $I$ never goes in head position, and then $\TT C_i[R]$ must be solvable for all terms $R$, or $\TT C_i[I]$ reduces to $\lambda \vec z . I M_1\dots M_p \to_\beta \lambda \vec z.  M_1 M_2\dots M_k$, the latter reducing to some head normal form. We then have that for $k_i>p$, $\TT C_i[P_k]$ reduces to $\lambda \vec z . P_k M'_1\dots M'_p \to_\beta \lambda \vec z.\lambda y_{p+1}.\dots \lambda y_{k}.\lambda x. x$, so is solvable.

%

Letting $k$ be any natural number greater than $ \max_i k_i$, we deduce that, from the convergence of all $\TT C_i[I]$, for $i=1,\dots,m$, it follows that also
$\TT C_i[P_k]$ converges for all $i=1,\dots, m$, and we conclude then that $P_k\in B_\delta^{\mathrm{ctx}}(I)$.
\end{description}
\end{proof}


Recalling that $D_\infty$ induces the theory $\C H^*$, the relation $p_{\mathrm{B\"ohm}}\sqsubset p_{\mathrm{Scott}}$ is in accordance with what happens with the corresponding $\lambda$-theories. By contrast,
 while $D_\infty$ and the contextual preorder both induce the $\lambda$-theory $\C H^*$, the first induces a $\lambda$-PPM which is \emph{finer} than the contextual partial metric.
As can be seen in the proof in the Appendix, the reason behind this is that, given terms $M \sqsubseteq_{\mathrm{ctx}} P$, there exists open
 $p_{\mathrm{Scott}}$-balls $B_\epsilon(P)$ whose elements all lie above $M$, while 
  $p_{\mathrm{ctx}}$ cannot define any such ball, since whether $M\leq Q$ cannot be tested by applying only \emph{finitely} many contexts to $Q$ (cf.~Remark \ref{rem:ctx}).

\section{Quantifying the Taylor Expansion}\label{sec6}

In this section we discuss the Taylor expansion of $\lambda$-terms \cite{ER, Regnier2006, Regnier2008}, a well-studied method that refines methods based on B\"ohm trees and Scott domains, by decomposing the non-linear behavior of a term into the \emph{linear} behavior of a set of simpler terms, called \emph{resource $\lambda$-terms}. 
Notably, several well-known properties of $\lambda$-terms (like e.g.~continuity and stability), which were originally established by topological and semantic methods, can be proved in a simpler, combinatorial way, via the Taylor expansion \cite{Barbarossa2019}.

The famous \emph{commutation theorem} \cite{Regnier2008} says that the Taylor expansion commutes with the construction of the B\"ohm tree, and shows that the associated $\lambda$-theories coincide. By presenting the Taylor expansion as an \emph{isometric} transformation, we add a quantitative flavor to this result, showing that also the corresponding notions of program similarity coincide.

%
%
%

\subparagraph{Resource terms and the Taylor expansion}

As we said, the Taylor expansion associates a $\lambda$-term with a set of terms, called \emph{resource terms}, with a linear operational semantics. The set $\Lambda^r$ of resource terms is defined by the grammar
$
t:= x\mid \lambda x.t \mid t\langle t,\dots, t\rangle
$, 
where $\langle t,\dots, t\rangle$ indicates a finite multiset of terms.
We define an order $\prec$ over resource $\lambda$-terms as the context closure of the relation $\emptyset \prec \langle t_1,\dots, t_n\rangle$. 
The operational semantics of resource terms replaces the standard $\beta$-rule with a linear monadic rule 
$\to_r$ that relates a redex $(\lambda x.t)\langle u_1,\dots, u_n\rangle$ with the set of terms 
$t[u_{\sigma(1)}/x_1,\dots, u_{\sigma(n)}/n]$, obtained by replacing each occurrence $x_i$ of $x$ in $t$ by the term $u_{\sigma(i)}$,
whenever $t$ contains exactly $n$ occurrences of $x$ and where $\sigma$ is any permutation in $\mathfrak S_n$. 
For example, the resource term $(\lambda x.x\langle x\rangle )\langle y,z\rangle$ reduces to the set of terms $\{y\langle z\rangle, z\langle y\rangle\}$ corresponding to the two possible ways of distributing $y,z$ across the two occurrences of $x$ in $x\langle x\rangle$. Instead, the resource term $(\lambda x.x\langle x\rangle)\langle y\rangle$ reduces to the empty set: as the single occurrence of $y$ cannot be duplicated, it does not suffice to replace all occurrences of $x$ in $x\langle x\rangle$. 
More generally, if $t$ contains a number of occurrences of $x$ different from $n$, then $(\lambda x.t)\langle u_1,\dots, u_n\rangle\to_r \emptyset$. Thanks to the impossibility of duplicating terms, linear reduction $\to_r^*$ is not only confluent, but also strongly normalizing (in linear time).

The \emph{Taylor expansion} of a $\lambda$-term $M$ is a set $\C T(M)\subseteq \Lambda^r$ defined inductively as
$\C T(x)= \{x \}$, 
$\C T(\lambda x.M)= \{\lambda x.t \mid t\in \C T(M)\}$ and
$\C T(MN)= \{ t\langle t_1,\dots, t_n \rangle\mid t\in \C T(M), x_n \in \BB N, t_1,\dots, t_n\in \C T(N)   \}$. 
For example, the Taylor expansion of $\lambda x.\lambda y.yx$ is composed of all resource terms of the form 
$\lambda x. \lambda y.y \langle x,\dots, x\rangle$. 
Since reduction is confluent and strongly normalizing, we can define the set $\mathrm{nf}(\C T(M))$ containing the normal forms of the resource terms in $\C T(M)$. 

The Taylor expansion extends to \emph{partial} $\lambda$-terms by letting $\C T(\bot)=\emptyset$. 
In this way, we can define the Taylor expansion of a B\"ohm tree $\alpha\in\mathrm{Ide}(\C A)$ by 
$
\C T(\alpha)= \bigcup\{\C T(A)\mid A\in \alpha \},
$. 
The aforementioned commutation theorem says then that $\C T(\C B(M)) = \mathrm{nf}(\C T(M))$; together with the injectivity of $\C T$ over B\"ohm trees (which is easily proved), this shows the equivalence of the $\lambda$-theory $\C B$ and the $\lambda$-theory generated by equating all closed terms whose Taylor expansions have the same normal form. 

%

We provide an alternative, topological, presentation of the Taylor expansion of B\"ohm trees. A natural choice would be to take the Scott topology induced by the resource term order $\preceq$. However, under this order, $\Lambda^r$ is not a dcpo: limits of directed sequences need not exist (as they would correspond, just like B\"ohm trees, to infinite terms).
This leads then to consider, just like for partial terms, the {completion} $\mathrm{Ide}(\Lambda^r)$ of $\Lambda^r$, which forms an algebraic dcpo. The elements of $\mathrm{Ide}(\Lambda^r)$ can be seen as possibly infinite resource terms, and the compact elements correspond to the finite ones, that is, to ordinary resource terms.

Recall that $\mathrm{Ide}(\C A)$ can be identified with the set of B\"ohm trees; the Taylor expansion can be presented in this setting as a map $\C T^*: \mathrm{Ide}(\C A) \to\C P(\mathrm{Ide}(\Lambda^r))$ defined by 
$
\C T^*(\alpha)=\mathrm{Ide}(\C T(\alpha)).
$
To see that it is well-defined, let us observe that $\C T(\alpha)\subseteq \Lambda^r$, so $\mathrm{Ide}(\C T(\alpha))\subseteq \mathrm{Ide}(\Lambda^r)$ is a set of ideals.
Notice that the set $\C T^*(\alpha)$ is \emph{closed} with respect to the Scott topology of $\mathrm{Ide}(\Lambda^r)$.

%
%
%
%
%

%
%
%
%
%

\subparagraph{Defining a metric on $\Lambda^r$}

We introduce a PUM on $\Lambda^r$ quantifying the order $\preceq$, which is essentially an adaptation of the tree partial metric. A normal resource term is of the form
$
t=\lambda x_1.\dots.\lambda x_n. x b_1\dots b_m
$, 
where each $b_i$ is a finite multiset $b_i=\langle t_i^1,\dots, t_i^{m_i} \rangle$.
The \emph{height} of a resource term $h(t)$ is defined recursively as
$
h(t)= \max_{ij}h(t_i^j) +1
$, 
where $t$ is as above. For any variable occurrence $z$ in $t$, we define its \emph{height in $t$} $h_t(z)$ as $h_t(z)=1$ if $z$ is as $x$ above, and as $h_t(z)=h_{t_i^j}(z)+1$ if the occurrence is in $t_i^j$.

For any normal resource term $t$ and $n\leq h(t)$, we define the resource term $t|_n$, corresponding to the "truncation" of $t$ at height $n$:
 if $h(t)\leq n$, then $t|_n=t$, and if $h(t)>n$, then we replace any subterm of $t$ of the form $ x b_1\dots b_m$, where $x$ is at height $n$, by
$x \emptyset\dots \emptyset$.
Observe that $h(t|_n)\leq n$ and $h(t|_n)=n$ holds whenever $h(t)\geq n$.

\begin{definition}[resource partial metric]
For any two resource terms $t,u\in \Lambda^r$, we define
$$
r(t,u):= \inf\{ 2^{-n}\mid h(t), h(u)\geq n \text{ and } t|_n \ = \ u|_n  \}.
$$
\end{definition}

By arguing similarly to the case of trees, it can be shown that $r$ is a PUM, and that the order $\leq_r$ coincides with $\preceq$. Notice that $r(t,t)=2^{-h(t)}$. 

%

\subparagraph{Lifting the metric to $\C P(\Lambda^r)$}
We now discuss how to lift the metric $r$ to subsets of $\Lambda^r$. 
A standard way to lift a metric $d$ from a set $X$ to its powerset $\C P(X)$ is via the \emph{Hausdorff lifting}
$
H_d(A,B)= \max\{
\sup\limits_{a\in A}\inf\limits_{b\in B}d(a,b),
\sup\limits_{b\in B}\inf\limits_{a\in A}d(a,b)
\}.
$ 
Intuitively, $H_d(A,B)$ looks, for each element of one set, for its \emph{closest} element in the other set, and then measures the distance that is obtained by this operation in the worst case. 
The same construction, when applied to a partial metric $p$, yields the \emph{partial Hausdorff metric} $H_p$ (see 
\cite{Aydi2012, Stubbe2018}) which, in spite of its name, is in fact \emph{not} a partial metric, as it satisfies a \emph{weaker} triangular law $H_p(A,B)\leq H_p(A,C)+H_p(C,B)-\inf_{c\in C}p(c,c)$. 

In any case, the Hausdorff lifting $H_r$ of the resource partial metric is not the right choice for us: suppose $\alpha$ is an infinite B\"ohm tree, so that its self-distance is $0$; then $\C T(\alpha)$ is a set of \emph{finite} terms of arbitrary depth, so that
$
H_r(\C T(\alpha),\C T(\alpha))= \sup_{t\in \C T(\alpha)}r(t,t)=\sup\{2^{-|t|}\mid t\in \C T(\alpha)\}=\frac{1}{2}>0=p_{\mathrm{tree}}(\alpha,\alpha)
$.
Beyond making the Taylor expansion non-isometric, from this we deduce that $H_r$ is constantly $\frac{1}{2}$ over \emph{all} non-empty Taylor expansions!

Instead, we introduce the following variant of the Hausdorff lifting:
\begin{definition}
For any PM $p:X\times X\to [0,1]$, let $H_p^*:\C P(X)\times \C P(X)\to [0,1]$ be:
$$
H^*_p(A,B)= \max\left\{
\sup\limits_{a\in A}\inf\limits_{
a'\geq_p a\in A, b\in B
}p(a',b), 
\sup\limits_{b\in B}\inf\limits_{
b'\geq_p b\in B, a\in A
}p(a,b')
\right\}.
$$
\end{definition}

Intuitively, on two sets $A,B$, $H^*_p(A,B)$ measures how close the elements of $A$ get to the elements of $B$ as soon as one is allowed to freely move higher within $A$ and $B$ following the order $\leq_p$.
Notice that, for $\alpha$ an infinite B\"ohm tree, we now have $H_r^*(\C T(\alpha),\C T(\alpha))=0$, as desired.
Similarly to the partial Hausdorff metric $H_p$, for a partial metric $p$, $H_p^*$ is \emph{not} in general a partial metric. Indeed, it only satisfies the following properties:
\begin{restatable}{proposition}{hausdorffone}
For any partial metric space $(X,p)$, the distance $H^*_p$ satisfies:
\begin{enumerate}
\item $H^*_p(A,A)\leq H^*_p(A,B)$;
\item $H^*_p(A,B)=H_p^*(B,A)$;
\item $H^*_p(A,B)\leq H^*_p(A,C)+H^*_p(C,B)-\inf_{c\in C}p(c,c)$.

\end{enumerate}
\end{restatable}
\longversion{
\begin{proof}
For point 1.~first observe that 
$
H_p^*(A,A)= \sup_{a\in A}\inf_{a'\geq a}p(a',a')$, 
and since $p(a',a')\leq p(a',b)$ holds for all $a',b$ we have
$$
H_p^*(A,A)= \sup_{a\in A}\inf_{a'\geq a}p(a',a')
\leq  \sup_{a\in A}\inf_{a'\geq a, b\in B}p(a',b)\leq H_p^*(A,B).
$$

Point 2.~follows from the definition.

Let us finally establish 3. From 
$$
p(a',b) \leq p(a',c) + p(c, b) - p(c,c)
$$
we first deduce
$$
\inf_b p(a',b) \leq  \inf_b \big ( p(a',c) + p(c, b) - p(c,c)\big )=
p(a',c) +   \big( \inf_b p(c, b)\big)  - p(c,c)
$$
Then, using the fact that $r(a',c)\geq r(a',c')$ whenever $c\leq c'$, we deduce 
$$
\inf_b p(a',b) \leq \inf_{c'\geq c} \big( p(a',c') + \inf_b p(c', b) - p(c',c')\big)\leq 
\inf_{c'\geq c} \big( p(a',c) + \inf_b p(c', b) - p(c',c')\big)
$$
and then, since $c$ is arbitrary,
\begin{align*}
\inf_b p(a',b) &\leq\inf_c \inf_{c'\geq c} \big( p(a',c) + \inf_b p(c', b) - p(c',c')\big)\\
&
=
\inf_c   p(a',c) + \inf_{c'\geq c}\inf_b\big (p(c', b) - p(c',c')\big)\\
& \leq
\inf_c   p(a',c) + \inf_{c'\geq c}\inf_b p(c', b) - \inf_c p(c,c)\\
& \leq
\inf_c   p(a',c) +\sup_{c''} \inf_{c'\geq c''}\inf_b p(c', b) - \inf_c rpc,c).
\end{align*}
From all this we now deduce 
$$
\sup_a \inf_{a'\geq a}\inf_b p(a',b)
\leq
\sup_a \inf_{a'\geq a} \inf_c   p(a',c) + \sup_{c''} \inf_{c'\geq c''}\inf_b p(c', b) - \inf_c p(c,c).
$$
which is precisely condition 4.

\end{proof}

}

However, $H_p^*$ is in fact a PM when restricted to $\mathrm{Ide}_p(X)$, the dcpo of ideals with respect to the order $\leq_p$. 

\longversion{
Let us first establish a preliminary lemma:

\begin{lemma}\label{lemma:finitedirected}
Let $A\subseteq \Lambda^r$ be directed and bounded, that is, such that, for some $k\in \BB N$ and for all $t\in A$,
$t$ has height at most $k$. Then there exists $\hat t\in A$ such that $A=\{u\mid u\leq t \}$. 
\end{lemma}
\begin{proof}
Consider an enumeration $t_n$ of the elements of $A$ and let $t^\dag\in A$ be an element of maximum height. 
By directedness of $A$, we obtain an increasing sequence $v'_n$ defined by $v'_0: t^\dag$ and $v'_{n+1}$ such that
$t_n,v'_n\leq v'_{n+1}$. Since the height of the $v'_n$ is constant, after a finite number of iterations the sequence $v'_n$ must become constant. This means that it reaches a maximum element $\hat t$ of $A$.
\end{proof}

}

\begin{restatable}{proposition}{Hrpm}
For any PM $p$ on $X$, $H^*_p$ is a PM on $\mathrm{Ide}_p(X)$ quantifying the order $\subseteq$.
\end{restatable} 
\longversion{
\begin{proof}
Let us first observe that for all non-empty $A\in \mathrm{Ide}_p(X)$, 
$$
H^*_p(A,A)= \sup_{a\in A}\inf\limits_{a'\geq a}p(a',a')= \inf_{a\in A}p(a,a).
$$
To prove this, let $a_1,a_2\in A$; then there exists $a_3\geq a_1,a_2$, and notice that $p(a_3,a_3)\leq p(a_1,a_1),p(a_2,a_2)$. So, on the one hand, from $a_3^\uparrow\subseteq a_2^\uparrow, a_1^\uparrow$ we deduce $\inf\limits_{a'\geq a_1}p(a',a'), \inf\limits_{a'\geq a_2}p(a',a')\geq 
\inf\limits_{a'\geq a_3}p(a',a')$; on the other hand, for any $c\in a_1^\uparrow$ we can find $d\in a_3^\uparrow\cap c^\uparrow$ (and similarly for $a_2$), so we also deduce $\inf\limits_{a'\geq a_1}p(a',a'), \inf\limits_{a'\geq a_2}p(a',a')\leq 
\inf\limits_{a'\geq a_3}p(a',a')$. 

We have thus shown that for any fixed $a_0\in A$, 
$$
\sup_{a\in A}\inf\limits_{a'\geq a}p(a',a')
= 
\inf\limits_{a'\geq a_0}p(a',a').
$$
It remains to show that $\inf\limits_{a'\geq a_0}p(a',a')=\inf_{a\in A}p(a,a)$: 
on the one hand, that $\inf\limits_{a'\geq a_0}p(a',a')\geq\inf_{a\in A}p(a,a)$ is clear; conversely, for all $a\in A$ there exists $a'\in a_0^\uparrow \cap a^\uparrow$, so that $p(a',a')\leq p(a,a),p(a_0,a_0)$, and we can thus conclude $\inf\limits_{a'\geq a_0}p(a',a')\leq\inf_{a\in A}p(a,a)$.

Using this fact, by applying the previous proposition we immediately deduce properties 1,2 and 4 of partial metric space. 
Let us now prove that $H_p^*$ captures the inclusion order. in other words:
$$
H_p^*(A,B)=H_p^*(A,A) \Leftrightarrow A\subseteq B.
$$
First observe that $H_p^*(B,B)\leq H_p^*(A,A)$, since for all $a\in A$ there exists $b\in B$ such that $b\geq a$.

Suppose $ A\subseteq B$ and let $b\in B$; then 
$$
\inf_{b'\geq b\in B}\inf_{a\in A}p(b,a)= \inf_{b'\geq b\in B}p(b,b)= H_p^*(B,B).
$$
Let now $a\in A$, then
$$
\inf_{a'\geq a\in A}\inf_{b\in B}p(a,b)= \inf_{a'\geq a\in A}p(a',a')= H_p^*(A,A).
$$
This implies then $H_p^*(A,B)=H_p^*(A,A) $. 

Conversely, suppose $H_p^*(A,B)=H_p^*(A,A)$: then $\sup_{a\in A}\inf_{a'\geq a\in A}\inf_{b\in B}p(a',b)\leq \inf_{a\in A}p(a,a)$, that is
$$
\forall a,a'\in A \exists a''\geq a\in A\exists b\in B \ \ p(a'',b)\leq p(a',a').
$$
Since $p(b,b)\leq p(a',b)$ this implies then
$$
\forall a\in A \exists b\in B \ \ p(b,b)\leq p(a,a),
$$
that is, $\forall a\in A\exists b\in B \ a\leq b$, and since $B$ is lower, this implies that also $a\in B$. We conclude then $A\subseteq B$. 

Since $H_p^*$ captures the reversed order inclusion, we deduce that also axiom 3 of partial metric spaces holds. 
\end{proof}

}

When $p=r$, the resource partial metric, $H_r^*$ indeed quantifies the Scott topology:

\begin{restatable}{proposition}{HrScott}
The PM $H_r^*$ quantifies the Scott topology on $\mathrm{Ide}_r(\Lambda^r)$.
\end{restatable}
\longversion{
\begin{proof}
Let us first show that open balls are Scott open.
It is clear that open balls are upper sets. Suppose now $C\in B_\epsilon(A)$, that is, 
$$
\sup_{c\in C}\inf_{c'\geq c\in C}\inf_{a\in A}r(c',a),
\sup_{a\in A}\inf_{a'\geq a\in A}\inf_{c\in C}r(c,a')< 
\inf_{a\in A}r(a,a)+\epsilon.
$$
We must find $C'\subseteq_{\mathrm{fin}}C$ such that 
$C'\in B_\epsilon(A)$.

First suppose $\inf_{a\in A}r(a,a)=0$, and let $k$ be maximum such that $2^{-k}\leq\epsilon$. 
Fix some $c\in C$; then there exists $c'\in C$ and $a'\in A$ such that $c\leq c'$ and $r(c',a')<2^{-k}$. 
In other words, $c'$ is a tree that coincides with $a'$ up to height $k$. 
Let $C'=\downarrow c'=\{ d\mid d\leq c'\}$. Notice that $C'$ is finite. We must show that $H_r^*(C',A)<2^{-k}$: 
on the one hand we have that for all $d\in C'$, $d\leq c'$ and there exists $a'\in A$ such that $r(c',a')<2^{-k}$;
on the other hand, for all $e\in A$, since $A$ is directed, there exists $e'\in A$ such that $a',e\leq e'$, and we have that
$r(e',c')\leq r(a',c')<2^{-k}$. We can thus conclude $H_r^*(C',A)<2^{-k}$.

Suppose new that $\inf_{a\in A}r(a,a)>0$ and let $k$ be such that $\inf_{a\in A}r(a,a)\leq 2^{-k}$. 
We deduce that $A$ is finite and bounded, so by Lemma \ref{lemma:finitedirected} $A=\downarrow \hat a$, for some maximum element $a\in A$ of height at most $k$. In particular, $\inf_{a\in A}r(a,a)=r(\hat a, \hat a)$. 
This implies in particular that $A$ is finite. 
Fix some $c\in C$; then there exists $c'\geq c$ such that $r(c',\hat a)<r(\hat a,\hat a)+\epsilon$. 
Then, taking $C':=\downarrow c'$ one can easily check as before that $H_r^*(C',A)<r(\hat a, \hat a)+\epsilon$.

For the converse direction, suppose $A\ll C$, that is, $A\subseteq_{\mathrm{fin}}C$. We must find $\epsilon>0$ such that $D\in B_\epsilon(C)$ implies $A\subseteq D$. 

Since $A$ is finite and directed, it has a maximum element $\hat a$, and since $A$ is downward closed, this implies $A=\downarrow \hat a$. 
Let $k$ be the height of $\hat a$ and let $0<\epsilon < 2^{-k+1}-2^{-k}$; if $D\in B_\epsilon(C)=B_\epsilon(A)$ 
then any $d\in D$ is below some $d'\in D$ such that $r(d',\hat a)<r(\hat a,\hat a)+\epsilon < 2^{-k+1}$, which implies that $d'$ coincides with $\hat a$ up to height $k$, and thus that $\hat a\leq d'$.
By downward closure, we deduce then that $\hat a\in D$, and thus that $A\subseteq D$. 
\end{proof}

}

\subparagraph{Taylor is an isometry}

%

The Taylor expansion can be presented either as a map $\C T:\Lambda\to \C P(\Lambda^r)$ turning a $\lambda$-term into a set of resource terms, or as a map $\C T^*:\mathrm{Ide}(\C A)\to \C P(\mathrm{Ide}(\Lambda^r))$ turning a B\"ohm tree (i.e.~an infinitary normal $\lambda$-term) 
into a set of infinitary resource terms.

We will show that both maps are isometries, when considering $\Lambda$ with the B\"ohm PM and $\mathrm{Ide}(\C A)$ with the tree PM, and measuring sets of (finite/infinite) resource terms via the lifting $H_r^*$ of the resource partial metric.

Let the $\lambda$-PPM $p_{\mathrm{Taylor}}$ be defined by 
$p_{\mathrm{Taylor}}(M,N)=H_r^*(\mathrm{nf}(\C T(M)), \mathrm{nf}(\C T(N)))$.
As we observed, the $\lambda$-theory generated by equating all terms $M,N$ such that $\mathrm{nf}(\C T(M))=\mathrm{nf}(\C T(N))$ coincides the theory $\C B$. Our result will extend this to the corresponding quantitative theories.


Let us first consider the Taylor expansion of $\lambda$-terms. 

\longversion{
We need to establish a few preliminary lemmas.

\begin{lemma}\label{lemma:zero}
For any infinite B\"ohm tree $T$, $H_r^*(\C T(\alpha), \C T(\alpha))=0$.
\end{lemma}
\begin{proof}
We have that 
$$H_r^*(\C T(\alpha), \C T(\alpha))= \sup_{t\in \C T(\alpha)}\inf_{u\in \C T(\alpha), u\succeq t}r(u,u)=\sup_{t\in \C T(\alpha)}\inf_n2^{-n}=0.$$
Indeed, since $T$ is infinite, for any $t\in \C T(\alpha)$ one can find $u\in \C C(\alpha)$ expanding $u$ arbitrarily deep, so that $r(u,u)$ can be made arbitrarily low.
\end{proof}

The following facts are easily established:
\begin{lemma}\label{lemma:height}
For any B\"ohm tree $\alpha$,
\begin{enumerate}
\item for all $t\in \C T(\alpha)$, $h(t)\leq |\alpha|$;
\item for all $n\leq |\alpha|$, and $t\in \C T(\alpha)$, there exists $t'\in \C T(\alpha)$ extending $t$ such that $h(t')\geq n$. 
\item for all $n\leq |\alpha|$, and $t\in \C T(\alpha)$, if $h(t)=n$, then $t\in \C T(\alpha_n)$.

\end{enumerate}
\end{lemma}

\begin{lemma}\label{lemma:erre}
For all B\"ohm trees $\alpha,\beta$, if $\alpha_n\neq \beta_n$, then there exists $t\in \C T(\alpha)$ such that for all  $t'\in \C T(\alpha)$ extending $t$ and for all $u\in \C T(\beta)$, $r(t,u)\geq 2^{-n}$.
\end{lemma}
\begin{proof}
Since $\alpha_n$ is defined, $\alpha$ is non-empty. Suppose $\alpha_n$ and $\beta_n$ differ at a certain position $\pi$ at height $n$ by $R$ (with $R$ being either the name of some head variable and/or the number of its arguments).
Take $t\in \C T(\alpha_n)$. By Lemma \ref{lemma:height} 2.~and 1.~there exists $t^\sharp\in \C T(\alpha_n)$ extending $t$ and such that $h(t)=n$. We can furthermore suppose that $t^\sharp$ is \emph{maximal}, that is, that for any $t'\in \C T(\alpha_n)$ extending $t^\sharp$, $t'=t^\sharp$.

Let now $t'\in \C T(\alpha)$ extending $t^\sharp$ and let 
 $u\in\C T(\beta)$; if $h(u)\leq n$, then $r(t', u)\geq 2^{-n}$ so we are done; suppose then $h(u)> n$ and suppose $r(t', u)=2^{-m}<2^{-n}$ (so that $m>n$). Then $t'|_m = u|_m$ which implies in particular $t^\sharp=t'|_n =u|_n$.
 However, this is impossible since the maximality of $t^\sharp$ implies that $t^\sharp$ coincides with $\alpha$ at position $\pi$, while $u|_n$ cannot coincide with $\alpha$ at position $\pi$. We conclude then that $r(t', u)\geq 2^{-n}$. 
\end{proof}

We can now proceed to the main argument.

}

\begin{restatable}{theorem}{isometryone}\label{prop:isometry}
$\C T: (\Lambda,p_{\mathrm{B\"ohm}})  \longrightarrow ( \C P(\Lambda^r), H_r^*)$ is an isometry.
Thus, $p_{\mathrm{Taylor}}= p_{\mathrm{B\"ohm}}$.
\end{restatable}
\longversion{
\begin{proof}
Let $\alpha,\beta\in \C B$. 
If $p_{\mathrm{tree}}(\alpha,\beta)=0$, then $\alpha=\beta$ is an infinite tree and by Lemma \ref{lemma:zero}, $H_r^*(\C T(\alpha), \C T(\beta))=
H_r^*(\C T(\alpha), \C T(\alpha))=0= p_{\mathrm{tree}}(\alpha,\beta)$.

Suppose that either $\alpha$ or $\beta$ is the empty tree, so that $p_{\mathrm{tree}}(\alpha,\beta)=1=2^0$. Say $\alpha$ is the empty tree, so $\C T(\alpha)=\emptyset$, and thus $H_{p_{\mathrm{tree}}}^*(\C T(\alpha), \C T(\beta))=H_{p_{\mathrm{tree}}}^*(\emptyset, \C T(\beta))=1=p_{\mathrm{tree}}(\alpha,\beta)$. This follows from the fact that the sup over the empty set is $0$, while the inf over the empty set is $1$, so their $\max$ is $1$.

Suppose now that $\alpha$ and $\beta$ are non-empty and $p_{\mathrm{tree}}(\alpha,\beta)=2^{-n}$. Then $\alpha_n=\beta_n$ while three possible cases occur: (1) $\alpha_{n+1}$ is defined while $\beta_{n+1}$ is not defined, (2) $\beta_{n+1}$ is defined while $\alpha_{n+1}$ is not defined, (3) $\alpha_{n+1}$  and $\beta_{n+1}$ 
are both defined and $\alpha_{n+1}\neq \beta_{n+1}$.

In all three cases we claim that:
\begin{itemize}
\item[i.] for any $t\in \C T(\alpha)$ there exists $t'\in \C T(\alpha)$ extending $t$ such that $h(t')\geq n$ and $t'\in \C T(\beta)$;
\item[ii.] for any $u\in \C T(\beta)$ there exists $u'\in \C T(\beta)$ extending $u$ such that $h(u')\geq n$ and $u'\in \C T(\alpha)$.

\end{itemize}
Claim i.~is proved as follows: by applying Lemma \ref{lemma:height} 2., one has that for any $t\in \C T(\alpha)$ there exists $t'\in \C T(\alpha)$ extending $t$ such that $h(t')\geq n$. Then, by applying Lemma \ref{lemma:height} 3., $t'\in \C T(\alpha_n)=\C T(\beta_n)$, when $t'\in \C T(\beta)$ holds as well. 
Claim ii.~is proved in a similar way.

Now, claims i.~and ii.~imply that
\begin{align*}
\sup_{t\in \C T(\alpha)}\inf_{t'\in \C T(\alpha), t'\succeq t}\inf_{u\in \C T(\beta)}r(t',u) &\leq 2^{-n},\\
\sup_{u\in \C T(\beta)}\inf_{u'\in \C T(\beta), u'\succeq u}\inf_{t\in \C T(\alpha)}r(t,u') &\leq 
 2^{-n}.
\end{align*}
We can thus conclude that $H_{p_{\mathrm{tree}}}^*(\C T(\alpha), \C T(\beta))\leq p_{\mathrm{tree}}(\alpha,\beta)$. 

%
%

To prove the converse direction we consider the three cases separately.

In case (1) we have that $|\beta|=n$. By Lemma \ref{lemma:height} 1.~together with Lemma \ref{lemma:height} 2., we have that for any $u\in \C T(\beta)$ there exists $u'\in \C T(\beta)$ extending $u$ and such that $h(u)=n$. 
Since $\beta$ is not the empty tree, $\C T(\beta)$ is non-empty and we deduce then that there exists $u\in \C T(\beta)$ such that $h(u)=n$. Since $|\beta|=n$, from Lemma \ref{lemma:height} 1.~we deduce that for any $u'\in \C T(\beta)$ extending $u$, it must be $u'=u$, and thus $r(u',u')= 2^{-n}$. This implies then
$$
\sup_{u\in \C T(\beta)}\inf_{u'\in \C T(\beta), u'\succeq u}\inf_{t\in \C T(\alpha)}r(t,u')\geq 
\sup_{u\in \C T(\beta)}\inf_{u'\in \C T(\beta), u'\succeq u}r(u',u')\geq 2^{-n}.
$$
From this we deduce $H_{p_{\mathrm{tree}}}^*(\C T(\alpha), \C T(\beta))\geq p_{\mathrm{tree}}(\alpha,\beta)$. 

For case (2) we can argue similarly to case (1).

Let us now consider case (3). By Lemma \ref{lemma:erre} we deduce that there exists $t\in \C T(\alpha)$ such that for all $t'\in \C T(\alpha)$ extending $t$ and for all $u\in \C T(\beta)$, $r(t',u)\geq 2^{-n}$. This claim implies then
$$
\sup_{t\in \C T(\alpha)}\inf_{t'\in \C T(\alpha), t'\succeq t}\inf_{u\in \C T(\beta)}r(t',u)\geq 2^{-n},
$$
and thus $H_{p_{\mathrm{tree}}}^*(\C T(\alpha), \C T(\beta))\geq p_{\mathrm{tree}}(\alpha,\beta)$, as desired.
\end{proof}

}

The results above states that, whenever the B\"ohm trees of two terms $M,N$ differ at height $n$, then,  
by moving higher and higher in their normalized Taylor expansions $\C T(M)$ and $\C T(N)$, one can find resource terms that differ precisely at height $n$, and can do no better.  

%

Let us now consider the map $\C T^*$. Since $\mathrm{Ide}(\Lambda^r)$ is quantified by $H_r^*$, we can consider its lifting $H^*_{H_r^*}$ to $\C P(\mathrm{Ide}_p(\Lambda^r))$. In fact, the computation of $H^*_{H_r^*}$ leads us back to $H_r^*$:

\begin{restatable}{lemma}{hausdorfftwo}
For all $\lambda$-terms $M,N$,
$H^*_r(\C T(M), \C T(N))= H^*_{H_r^*}(\C T^*(M), \C T^*(N)).$
\end{restatable}
\longversion{

\begin{proof}
Suppose that $H^*_r(\C T(M), \C T(N))\leq \epsilon$. In other words, that 
\begin{align*}
\sup_{t\in \C T(M)}\inf_{t'\geq t\in \C T(M)}\inf_{u\in \C T(N)}r(t',u)& \leq \epsilon,\\
\sup_{u\in \C T(N)}\inf_{u'\geq u\in \C T(N)}\inf_{t\in \C T(M)}r(t,u')& \leq \epsilon.
\end{align*}
Let $A\in \mathrm{Ide}(\C T(M))$ and let $t\in A$; then there exists $t'\geq t$ and $u\in \C T(N)$ such that $r(t',u)\leq \epsilon$; let $A'\supseteq A$ be such that $A'\subseteq \C T(M)$ and $A'$ contains $t'$; moreover, let $B=u^{\downarrow}\in \mathrm{Ide}(\C T(N))$; we can check that 
$H_r^*(A',B)\leq \epsilon$: on the one side we have 
$$
\forall v\in A' \exists v'\geq v\in A \exists w\in B \ r(v',w)\leq \epsilon,
$$
and
$$
\forall w\in B \exists w'\geq w\in B \exists v\in A' \ r(v,w')\leq \epsilon.
$$
We have thus proved that 
$$\sup_{A\in \C T^*(M)}\inf_{A'\supseteq A\in \C T^*(M)}\inf_{B\in \C T^*(N)}H_r^*(A',B)\leq \epsilon.$$ 
By a dual argument we can then show
$$\sup_{B\in \C T^*(N)}\inf_{B'\supseteq B\in \C T^*(N)}\inf_{A\in \C T^*(M)}H_r^*(A,B')\leq \epsilon.$$
We can thus conclude that $H^*_{H_r^*}(\C T(M), \C T(N))\leq \epsilon$ and thus that 
$H^*_r(\C T(M), \C T(N))\geq H_{H_r^*}(\C T^*(M), \C T^*(N))$.

Suppose now $ H_{H_r^*}(\C T^*(M), \C T^*(N))\leq \epsilon$, that is 
\begin{align*}
\sup_{A\in \C T^*(M)}\inf_{A'\supseteq A\in \C T^*(M)}\inf_{B\in \C T^*(N)}H_r^*(A',B)&\leq \epsilon,\\
\sup_{B\in \C T^*(N)}\inf_{B'\supseteq B\in \C T^*(N)}\inf_{A\in \C T^*(M)}H_r^*(A,B')&\leq \epsilon.
\end{align*}
Let $t\in \C T(M)$, letting $A=t^\downarrow$, there exists $A'\supseteq A\in \C T^*(M)$ and 
$B\in \C T^*(N)$ such that 
$$
\forall v\in A'\exists v'\geq v\in A' \exists w\in B \ r(v',w)\leq \epsilon.
$$
Since $t\in A'$ we deduce then 
$$
\exists t'\geq t\in \C T(M) \exists u\in \C T(N) \ r(t',u)\leq \epsilon,
$$
that is
$$
\inf_{t'\geq t\in \C T(M)} \inf_{u\in \C T(N)}  r(t',u)\leq \epsilon.
$$
By a dual argument, we can show that for all $u\in \C T(N)$,
$$
\inf_{u'\geq u\in \C T(N)} \inf_{t\in \C T(M)}  r(t,u')\leq \epsilon.
$$
and we can conclude $H_r^*(\C T(M), \C T(N))\leq \epsilon$, whence \\
$H^*_r(\C T(M), \C T(N))\leq H_{H_r^*}(\C T^*(M), \C T^*(N))$.
\end{proof}

}

Thanks to Proposition \ref{prop:isometry}, this immediately produces:
\begin{restatable}{theorem}{isometrytwo}
$\C T^*:(\mathrm{Ide}(\C A),p_{\mathrm{tree}})\longrightarrow (\C P(\mathrm{Ide}(\Lambda_r)), H^*_{H^*_r})$ is an isometry.
\end{restatable}
\longversion{
\begin{proof}
Follows from Proposition \ref{prop:isometry} and the Lemma above.
\end{proof}
}
 
\shortversion{
\begin{remark}
As shown in detail in the long version, we can obtain an isometry also if we choose to measure B\"ohm trees and Taylor expansions using the PMs from Examples \ref{ex:altbohm} and \ref{ex:altpower}. Indeed, for any enumeration $(A_n)_n$ of partial terms, one can define an enumeration $(r_n)_n$ of resource terms and weights $\theta_n$ such that
$\C T:(\C B,p^{\C B}_{(A_n)_n,\frac{1}{2^n}})\longrightarrow (\C P(\Lambda_r), p^{\C P}_{(r_n)_n,\theta_n})$ is an isometry.
\end{remark}
}

\longversion{
We can obtain an isometry also if we choose to measure B\"ohm trees and Taylor expansions using the PMs from Examples \ref{ex:altbohm} and \ref{ex:altpower}. Indeed, for any enumeration $(A_n)_n$ of partial terms, one can define an enumeration $(r_n)_n$ of resource terms and weights $\theta_n$ such that
$\C T:(\C B,p^{\C B}_{(A_n)_n,\frac{1}{2^n}})\longrightarrow (\C P(\Lambda_r), p^{\C P}_{(r_n)_n,\theta_n})$ is an isometry.

\begin{restatable}{theorem}{isometrythree}
For any enumeration $(A_n)_n$ of partial terms there exists an enumeration $(r_n)_n$ of resource terms such that
$\C T:(\mathrm{Ide}(\C A),p^{\C B}_{(A_n)_n})\longrightarrow (\C P(\Lambda_r), p^{\C P}_{(r_n)_n})$ is an isometry.
\end{restatable}
%
\begin{proof}
Let us define a relation $t \vartriangleleft A$ between resource $\lambda$-terms and partial terms via the following rules:
$$
\AXC{}
\UIC{$x\vartriangleleft x$}
\DP
\qquad 
\AXC{$t\vartriangleleft A$}
\UIC{$\lambda x.t \vartriangleleft \lambda x. A$}
\DP
\qquad
\AXC{$\Big( t^i_1,\dots, t^i_{m_i}\vartriangleleft M_i \Big )_{i=1,\dots, n}$}
\UIC{$x[\vec t^1]\dots [\vec t^n]\vartriangleleft xM_1\dots M_n $}
\DP
$$
One can then check, by induction on partial terms, that $\C T(A)=\{t\mid t\vartriangleleft A\}$.

For any term $A$, let us fix an enumeration $t_{A,n}$ of the resource terms $t$ such that $t\vartriangleleft A$. 
Given a bijection $q: \BB N\to \BB N^2$, define an enumeration $v_n:= t_{q_1(n), q_2(n)}$ and a weighting $\theta_n:= \frac{1}{2^{q_1(n)+q_2(n)}}$. 
Indeed, we have that 
\begin{align*}
\sum_{n=1}^\infty\theta_n& =
\sum_{n=1}^\infty \frac{1}{2^{q_1(n)+q_2(n)}}=
\sum_{n=1}^\infty \frac{1}{2^{q_1(n)}}\cdot \frac{1}{2^{q_2(n)}}
= \sum_{m,n=1}^\infty\frac{1}{2^{m}}\cdot \frac{1}{2^{n}}
\\
&=
\sum_{m=1}^\infty  \frac{1}{2^m}\sum_{n=1}^\infty \frac{1}{2^n}=\sum_{m=1}^\infty  \frac{1}{2^m}=1.
\end{align*}

Recalling that $\C T(\alpha)=\bigcup_{A\leq \alpha}\C T(A)=\bigcup_{A\leq \alpha}\{ t \mid t \vartriangleleft A\}$, we have that $t\in \C T(\alpha)$ iff $t \vartriangleleft A$ for some $A\leq \alpha$. We thus have 
\begin{align*}
p_{v_n, \theta_n}^{\C P}(\C T(\alpha),\C T(\beta))&=
\sum\left \{
\frac{1}{2^n}\cdot \frac{1}{2^m} \ \Big \vert \ 
t_{A_n, m}\notin \C T(\alpha) \text{ or } t_{A_n, m}\notin \C T(\beta)
\right\}\\
&=
\sum\left \{
\frac{1}{2^n}\cdot \frac{1}{2^m} \ \Big \vert \ 
A_n\not\leq \alpha \text{ or } A_n\not\leq \beta
\right\}\\
&=
\sum\left \{
\frac{1}{2^n}\cdot\sum_{m=1}^\infty \frac{1}{2^m} \ \Big \vert \ 
A_n\not\leq \alpha \text{ or } A_n\not\leq \beta
\right\}\\
&=
\sum\left \{
\frac{1}{2^n} \ \Big \vert \ 
A_n\not\leq \alpha \text{ or } A_n\not\leq \beta
\right\}= p^{\C B}_{A_n,\frac{1}{2^n}}(\alpha,\beta).
\end{align*}
\end{proof}

}

\section{Conclusions}\label{sec7}

\subparagraph{Related Work}

Since their introduction in \cite{matthews}, the literature on partial metrics has grown vast, and comprises both theoretical investigations \cite{Schellekens2004, Myronyk2022, Aydi2012, Jager2018} and connections with theoretical computer science \cite{Valero2011}, notably domain theory \cite{Bukatin1997, ONeill, Schellekens2003, Smyth2006}. Recently, an elegant categorical description of partial metric spaces as quantaloid-enriched categories has been proposed \cite{Stubbe2018}, as well as  a characterization of the partial metric spaces that are \emph{exponentiable} (in a category whose morphisms are the non-expansive - or 1-Lipschitz - functions and not, as in this paper, all continuous functions).
While, as we have said, the metrizability of Scott domains via partial metrics has been well known since \cite{Bukatin1997, ONeill}, not much is found in this vast literature about the specific use of partial metrics for studying the topological semantics of the $\lambda$-calculus or, more generally, of higher-order programming languages. 

Beyond partial metrics, the literature on higher-order program metrics has been growing vast as well. As the category Met of metric spaces and non-expansive functions is \emph{not} cartesian closed, the literature has focused on two complementary directions: on the one hand, restrict to cartesian closed \emph{sub}-categories of Met, like \emph{ultra}-metric spaces \cite{Escardo1999}, or \emph{injective} metric spaces \cite{Clementino2006}; \cite{Honsell2022} adapts Mardare's et al.'s quantitative equational theories \cite{Plotk} to higher-order languages, introducing a notion of \emph{quantitative $\lambda$-theory} (which, contrarily to $\lambda$-PPMs, require contexts to be non-expansive).

On the other hand, restrict attention to \emph{linear} \cite{Dahlqvist2022, Hoshino2023} or \emph{graded} \cite{Reed_2010, Gaboardi2017} $\lambda$-calculi, which can be modeled in Met. Notably, \cite{Gaboardi2017} introduces \emph{metric CPO}s, that is CPOs endowed with \emph{sub}-continuous metrics (i.e.~satisfying $d(\lim_n x_n,\lim_n y_n)\leq\epsilon$ whenever $d(x_n,y_n)\leq \epsilon$ holds for all $n$). This is a weaker condition than quantifiability, since the limits in the metric need not coincide with the CPO limits. 

Differential logical relations \cite{dallago, dallago2} have been recently introduced as a generalized approach to program metrics, relaxing usual Lipschitz, and even continuity, conditions. Notably, related models based on \emph{generalized} partial metric spaces are studied in \cite{Geoffroy2020, PistoneLICS}. In such models distances need not be positive reals but are computed on an arbitrary \emph{quantale}.

Finally, several works have investigated infinitary $\lambda$-calculi defined via a \emph{metric completion} of ordinary terms \cite{Kennaway, Mazza2013}.
These approaches are based on ultrametrics akin to the tree metric considered in this paper for B\"ohm trees. Recall that ordinary metric spaces are topologically Hausdorff, contrarily to the spaces considered in this paper. 
The metric completion of partial metric spaces is discussed in \cite{Ge2015, Stubbe2018}.

%
%
%
%
%
%
%
%
%
%

\subparagraph{Future Work}

While this paper focuses on metric counterparts for well-known techniques, our results suggest several potential developments.

The metrizability of Scott domains suggests to study models based on Lipschitz-continuous, rather than just continuous, functions, as is standard in the literature on linear $\lambda$-calculi. For instance, considering the B\"ohm metric, a non-expansive context should respect {depth}: if two terms $M,N$ coincide up to depth $n$, then $\TT C[M]$ and $\TT C[N]$ must also coincide up to depth $n$. This suggests connections with recent work on \emph{stratified} notions of program equivalence \cite{Arrial2024}.

Sections 4 and 6 introduced several methods to lift a partial metric to the powerset; using such liftings, as we suggest at several places, our results based on Scott domains could be adapted to the relational model, in which $\lambda$-terms are interpreted via relations $R\in\C P(A\times B)$.

While we here just considered the untyped $\lambda$-calculus and basic cartesian closed structure (i.e.~finite products and exponentials), the applicative distances introduced in this paper should adapt well also to dependently typed languages; moreover, our results on the Hausdorff lifting suggests that other monadic liftings (e.g.~the probability monad) could be considered. 
At the same time, the metric account of RealPCF suggested at in Example \ref{ex:realpcf} could be explored in more depth, for instance considering the behavior of operators like the parallel if or even program derivatives.

Finally, the fact that several partial metrics considered in this paper produce computable distances between finite approximants 
 suggests to explore potential connections with quantitative type systems related to the relational and topological semantics, like those based on non-idempotent intersection types \cite{Bucciarelli2017}.

\bibliography{main.bib}

\longversion{
\appendix

\section{The Model $D_\infty$}
In this appendix we shortly recall the construction of the domain $D_\infty$ and we establish a few properties which are used in the proof of Theorem \ref{th:dscott}.

Let us fix some domain $D$.
For each $n$, let the domain $D_{n}$ be defined by $D_{0}:= D$  and $D_{n+1}:= \C C(D_{n}, D_{n})$.
We define injection-projection pairs $(i_{n},j_{n}): D_{n} \to_{ip}D_{n+1}$ by
$i_{0}(x):= \lambda y.x $, $ j_{0}(f)= f(\bot) $, $i_{n+1}(f)= i_{n}\circ f\circ j_{n}$ and $j_{n+1}(g)= j_{n}\circ g\circ i_{n}$.

The domain $D_{\infty}$ is defined as 
$
D_{\infty}=\left \{ x \in \prod_{n = 0}^{\infty} D_{n} \  \Big \vert \  \forall n\in \BB N, x_{n}=j_{n}(x_{n+1})\right\}
$

For all $m<n$ (say $n=m+q$) we define maps $i_{mn}:D_m\to D_n$ and
$j_{nm}:D_n\to D_m$ via
\begin{align*}
	i_{m(m+1)}&:= i_m &   j_{(m+1)m}&:= j_m \\
	i_{m(m+q+2)}&:= i_{(m+1)(m+1+q)}\circ i_m &
	j_{(m+q+2)m}&:=
	j_{(m+1)m}\circ j_{(m+q+2)(m+1)}
\end{align*}
The behavior of the maps $i_{nm}$ and $j_{mn}$ is illustrated in Fig.~\ref{fig:diagram}. 
Moreover, we can define maps
$i_{n\infty}:D_n \to D_\infty$ and 
$j_{\infty n}:D_\infty \to D_n$ via
\begin{align*}
	(i_{n\infty}(x))_k&:=
	\begin{cases}
		j_{nk}(x)
		&
		n=k+q
		\\
		x
		&
		k=n
		\\
		i_{nk}(x)
		& k= n+q
	\end{cases}
\end{align*}
and
$
j_{\infty n}(x):= x_n.
$
Fig.~\ref{fig:diagram2} illustrates the behavior of the maps $i_{kn},j_{kn}$ and $i_{k\infty}, j_{\infty k}$ with respect to some element $x\in D_k$. 

$D_\infty$, together with the maps $i_{n\infty}$, can be shown to be the direct limit of the system generated by the maps $i_{mn}$.  
The diagram in Fig.~\ref{fig:diagram} illustrates the situation.
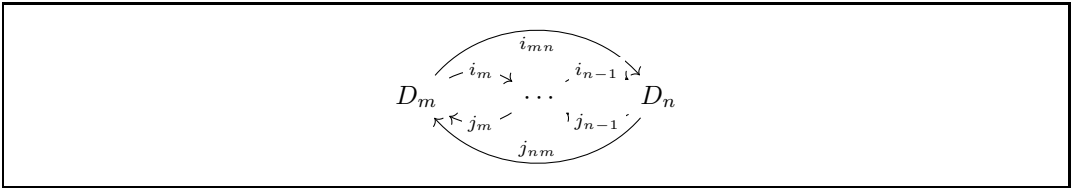
\begin{figure}
\fbox{
\begin{minipage}{0.98\textwidth}
$$
\begin{tikzcd}[ampersand replacement=\&]
	D_m \arrow[r, "i_{m}" description, bend left] \arrow[rr, "i_{mn}"', bend left=49] \& \cdots \arrow[r, "i_{n-1}" description, bend left] \arrow[l, "j_{m}" description, bend left] \& D_n \arrow[l, "j_{n-1}" description, bend left] \arrow[ll, "j_{nm}"', bend left=49]
\end{tikzcd}
$$
\end{minipage}
}
\caption{Illustration of the maps $i_{nm}$ and $j_{mn}$. Recall that a $i$ arrow followed by a $j$ arrow produces the identity, while a $j$ arrow followed by a $i$ arrow produces something smaller than the identity.}
\label{fig:diagram}
\end{figure}

\begin{figure}
\fbox{
\begin{minipage}{0.98\textwidth}
$$
\begin{tikzcd}[ampersand replacement=\&, scale cd=0.7]
	x_\infty = \& \cdots\& j_{km}(x) \in D_{m} \& \cdots \& x \in D_k \arrow[d, "i_{k\infty}", bend left] \& \cdots \& i_{kn}(x) \in D_{n}\& \cdots \\
	\&        \&                     \&        \& D_{\infty} \arrow[u, "j_{\infty k}", bend left]           \&        \&                     \&       
\end{tikzcd}
$$
\end{minipage}
}
\caption{The maps $i_{kn}, j_{kn}, i_{k\infty}, j_{\infty k}$ can be used to move an element $x\in D_k$ across all spaces $D_n$ for $n\in \BB N\cup\{\infty\}$.}
\label{fig:diagram2}
\end{figure}
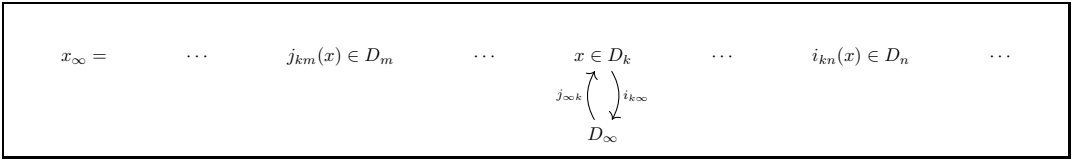

Observe that in $D_{\infty}$ we have that $x\ll y$ iff there exists $N\in \BB N$ such that 
$
x \leq \bigvee_{n\leq N}i_{n\infty}(y_{n})\ll y$.
This is an immediate consequence of the remark that $y=\bigvee_{n=0}^{\infty}i_{n\infty}(y_{n})$.

We now establish a few properties of this construction.

\begin{lemma}\label{lemmall1}
	For all $n\in \BB N$ and $x,y\in X_n$, $x\ll y$ implies $i_n(x)\ll i_n(y)$.
\end{lemma}
\begin{proof}
	Let $D\subseteq X_{n+1}$ be a directed set such that 
	$i_n(y)\leq \bigvee D$. 
	Let $jD=\{ j_n(d)\mid d\in D\}\subseteq X_n$. Let us show that $jD$ is directed: 
	given $j_n(d_1), j_n(d_2)\in jD$, $i_n(j_n(d_1))\leq d_1$ and 
	$i_n(j_n(d_2))\leq d_2$; now, since $D$ is directed, there exists $d_3\in D$ with 
	$d_1,d_2\leq d_3$; observe that this implies
	$j_n(d_1), j_n(d_2)\leq j_n(d_3)$. Since $j_n(d_3)\in jD$ we can conclude that $jD$ is directed.
	
	From $i_n(y)\leq \bigvee D$ we deduce
	that 
	$y=j_n(i_n(y))\leq j_n(\bigvee D)= \bigvee jD$, where in the last step we use the fact that $ j_n$ is continuous. Now, from $x\ll y$ it follows that $ x\leq j_n(d)$ holds for some $d\in D$, and thus
	$i_n(x)\leq i_n(j_n(d))\leq d$. We can thus conclude that 
	$i_n(x)\ll i_n(y)$.
\end{proof}

\begin{lemma}\label{lemmall22}
	For all $n\in \BB N$, $x\in X_n$ and $y\in X_\infty$, 
	$x \ll y_n$ $\To$ $i_{n\infty}(x)\ll y$.
\end{lemma}
\begin{proof}
	
	Suppose $x\ll y_n$. Let us first show that 
	$i_{n\infty}(x)\leq y$.
	First, from $x\leq y_n$ we deduce $i_{n(n+q)}(x)\leq i_{n(n+q)}(y_n)=
	i_{n(n+q)}(j_{(n+q)n}(y_{n+q}))
	\leq y_{n+q}$.
	Moreover, from $x\leq y_n$ we also deduce
	$j_{n(n-q)}(x)\leq j_{n(n-q)}(y_n)=y_{n-q}$. Hence, for all $q\in \BB N$, 
	$(i_{n\infty}(x))_q\leq y_q$, which implies $i_{n\infty}(x)\leq y$.

	Suppose now $D\subseteq X_\infty$ is a directed set such that $y\leq \bigvee D$. Let $D_n=\{d_n\mid d\in D\}$. Recall that $d_n=j_{\infty n}(d)$. 
	Let us show that $D_n$ is directed: let $j_{\infty n}(d_1),j_{\infty n}(d_2)\in D_n$, then 
	$i_{n\infty}(j_{\infty n}(d_1))\leq d_1$ and 
	$i_{n\infty}(j_{\infty n}(d_2))\leq d_2$ and since $D$ is directed there exists $d_3\in D$ such that $d_1,d_2\leq d_3$. Then we have that 
	$j_{\infty n}(d_1),j_{\infty n}(d_2)\leq j_{\infty n}(d_3)$ and, since 
	$ j_{\infty n}(d_3)\in D_n$, this shows that $D_n$ is directed as well.

	Now, since $d_n=j_{\infty n}(d)$ and $j_{\infty n}$ is continuous, 
	from $y\leq \bigvee D$ we deduce
	$y_n \leq j_{\infty n}(\bigvee D)= \bigvee D_n$. 
	From $x\ll y_n$ we conclude then that there exists $ j_n(d)\in D_n$ such that 
	$x\leq j_n(d)=d_n$. 
	
	To prove that $  i_{n\infty}(x)\leq d$, let us show that for all $k\in \BB N$, 
	$( i_{n\infty}(x))_k\leq d_k$: if $k=n$,
	$( i_{n\infty}(x))_k = x\leq d_n=d_k$; if $k< n$, then 
	$( i_{n\infty}(x))_k = j_{nk}(x)$ and from 
	$x\leq d_n$ we deduce 
	$( i_{n\infty}(x))_k=j_{nk}(x)\leq j_{nk}(d_n)=d_k$; finally, if $k>n$, 
	$( i_{n\infty}(x))_k = i_{nk}(x)$ and from 
	$x\leq d_n$ we deduce 
	$( i_{n\infty}(x))_k=i_{nk}(x)\leq i_{nk}(d_n)=i_{nk}(j_{kn}(d_k))\leq d_k$.
\end{proof}

\begin{lemma} \label{lemmall3}
	For all $x,y\in X_\infty$, $x\ll y$ iff there exists $N\in \BB N$ and 
	$w_1,\dots,w_k\in X_{N}$ such that
	$
	w_1,\dots, w_k \ll y_N
	$
	 and
	$
	x \leq  i_{N\infty}(w_1\vee\dots\vee w_k)
	$.
\end{lemma}
\begin{proof}
	Let $D$ be the set of all $i_{n\infty}(w_n)\in X_n$, for some $n\in \BB N$, such that $
	w_n \ll y_n$.
	Observe that, by Lemma \ref{lemmall22}, any $i_{n\infty}(w_n)\in D$ satisfies $i_{n\infty}(w_n)\leq y$. This implies that all finite joins of elements of $D$ exist in $X_\infty$.
	Let $D^*$ be the set of all such finite joins. Notice that $D^*$ is directed since $d,d'\in D^*$ implies $d\vee d^*\in D^*$. 
	
	Since a base for $X_\infty$ is formed by elements of the form $i_{n\infty}(b)$, here $b$ is in the base of $X_n$, and since $X_\infty$ is a continuous domain, we deduce that 
	$y\leq \bigvee D\leq \bigvee D^*$. 
	
	Now, from $x\ll y\leq \bigvee D^*$, since the $D^*$ is directed, we deduce that there exists finitely many elements $u_1, \dots, u_k$ of $D$, where $u_j\in X_{n_j}$, such that 
	$x \leq i_{n_1\infty}(u_1)\vee\dots\vee i_{n_k\infty}(u_k)$. 
	To conclude it suffices to let $N=\max\{n_j\mid j=1,\dots,k\}$ and let
	$w_j:= i_{n_jN}(u_j)$, by observing that 
	(1) $u_j\ll y_{n_j}$ implies $w_j\leq y_N$ by Lemma \ref{lemmall1} and (2) 
	$ i_{n_j\infty}(u_j)=
	i_{N\infty}(i_{n_jN}(u_j))=i_{N\infty}(w_j)$.
\end{proof}

\begin{lemma}\label{lemmall2}
	For all $n\in \BB N$, $x\in X_n$ and $y\in X_\infty$, 
	$i_{n\infty}(x)\ll y$ $\To$ $\exists N\forall k\geq N, i_{n(n+k)}(x)\ll y_{n+k}$;
\end{lemma}
\begin{proof}
	
	From $i_{n\infty}(x)\ll y$ via Lemma \ref{lemmall3} we deduce that there exists 
	$M\in \BB N$ and $w_1, \dots, w_k\in X_{N}$ such that 
	$w_j\ll y_M$ and 
	$i_{n\infty}(x) \leq  i_{M\infty}(w_1\vee\dots\vee w_k)$, where we can suppose that 
	$n\leq M$.

	Let $N=M-n$. 
	Since $i_{n\infty}(x)=i_{M\infty}(i_{nM}(x))$, from 
	$i_{n\infty}(x) \leq  i_{M\infty}(w_1\vee\dots\vee w_k)$ we deduce
	$i_{n(n+N)}(x)= (i_{n\infty}(x))_M\leq w_1\vee\dots\vee w_k\ll y_{M}$.

	Finally, to conclude that $i_{n(n+k)}(x)\ll y_{n+k}$ holds for all $k\geq N$ we can argue by induction, using Lemma \ref{lemmall1} and the fact that $i_{m}(y_m)\leq y_{m+1}$. 
\end{proof}

}
\end{document}